\documentclass[10pt,twocolumn]{IEEEtran}
\usepackage{pifont}
\makeatletter
\def\ps@headings{%
\def\@oddhead{\mbox{}\scriptsize\rightmark \hfil \thepage}%
\def\@evenhead{\scriptsize\thepage \hfil \leftmark\mbox{}}%
\def\@oddfoot{}%
\def\@evenfoot{}}
\makeatother
\usepackage[top=1in, bottom=1in, left=1in, right=1in]{geometry}
\usepackage{amsfonts}
\usepackage{amsmath}
\usepackage{amsfonts}
\usepackage{bbm}
\usepackage{subfigure}
\usepackage{epsfig}
\usepackage{amssymb}
\usepackage{mathrsfs}
\usepackage{algorithmic}
\usepackage[ruled]{algorithm}
\usepackage{footmisc}
\usepackage{xspace}
\usepackage{graphicx}
\usepackage{times}

\newtheorem{theorem}{Theorem}

\newtheorem{lemma}[theorem]{Lemma}
\def\ie{\textit{i.e.}\xspace}

\def\whp{{\emph{w.h.p}}\xspace}

\def\clockshift{t_{\delta}}

\def\colorc{c}
\def\colorset{\mathscr{C}}
\def\colorindex{\theta}
\def\communicateset{\vartheta}
\def\conone{c_1}

\def\coninterf{c_3}
\def\confour{c_5}
\def\datarate{\varrho}
\def\demand{D}
\def\degree{d}
\def\degreemax{\Delta}
\def\graphr{R}
\def\interfhop{\eta}

\def\interfpara{\varphi}
\def\interferenceM{\textbf{I}}
\def\load{\mathcal {L}}
\def\levelmax{\mathscr{L}}
\def\level{l}

\def\period{T}
\def\periodsize{m}
\def\protocol{\textbf{ADC}\xspace}
\def\quorum{Q}
\def\QS{\Omega}
\def\QSsizemin{c_4}
\def\realtime{t}
\def\region{\sigma}
\def\rotate{\mathscr{S}}
\def\selecttime{K}
\def\sink{s}
\def\slot{\tau}
\def\slotset{\mathfrak{L}}
\def\slotsetsize{\kappa}
\def\strategy{\mathcal {S}}
\def\subperiod{\varsigma}

\def\tree{\mathscr{T}}

\include{figures}
\allowdisplaybreaks %%% this make pages look tight.

\begin{document}
\title{Quorum-based Localized Scheme for Duty Cycling in Asynchronous Sensor Networks
%\thanks{This work is supported by the National Natural Science Foundation of China under Grants
%No.61003298 and 60803126, the sub-foundation of Zhejiang Provincial Key Innovation Team on Sensor Networks under
%Grant No.2009R50046-7 and the open funding of Science and Technology on Sonar Laboratory in Hangzhou City under Grant No. KF201103.}
}
%\title{Energy-aware Joint Data Aggregation and Link Scheduling under Harvested-Energy Imbalance in Solar Sensor Network}
\author{
\IEEEcompsocitemizethanks{This paper was published on IEEE 8th International Conference on Mobile Ad hoc and Sensor Systems (MASS, 2011)}\authorblockN{Jianhui Zhang\authorrefmark{1}\authorrefmark{3}, Shao-jie Tang\authorrefmark{2}, Xingfa Shen\authorrefmark{1}, Guojun Dai\authorrefmark{1}, Amiya Nayak\authorrefmark{3}}
\\\authorblockA{\authorrefmark{1}\small College of Computer Science and Technology, Hangzhou Dianzi
University, Hangzhou 310018 China.}
\\\authorblockA{\authorrefmark{2} Department of Computer Science, Illinois Institute of Technology, Chicago, IL60616, USA}
\\\authorblockA{\authorrefmark{3} SITE, University of Ottawa, Ottawa, Ontario K1N 6N5, Canada}\vspace{-3.5em}
}
\maketitle
\vspace{-4em}
\begin{abstract}
\label{section:abstract}
Many TDMA- and CSMA-based protocols try to obtain fair channel access and to increase channel utilization.
It is still challenging and crucial in Wireless Sensor Networks (WSNs), especially when the time synchronization cannot be well guaranteed and consumes much extra energy.
This paper presents a localized and on-demand scheme \protocol to adaptively adjust duty cycle based on \emph{quorum systems}.
\protocol takes advantages of TDMA and CSMA and guarantees that (1) each node can fairly access channel based on its demand; (2) channel utilization can be increased by reducing competition for channel access among neighboring nodes;  (3) every node has at least one rendezvous active time slot with each of its neighboring nodes even under asynchronization.
The latency bound of data aggregation is analyzed under \protocol to show that \protocol can bound the latency under both synchronization and asynchronization.
We conduct extensive experiments in TinyOS on a real test-bed with TelosB nodes to evaluate the performance of \protocol.
Comparing with B-MAC, \protocol substantially reduces the contention for channel access and energy consumption, and improves network throughput.
\end{abstract}
\IEEEpeerreviewmaketitle

%%%%%%%%%%%%%%%%%%%%%%%%%%%%%%%%
\section{Introduction}
\label{section:introduction}
%\input{introduction.tex}
%In Wireless Sensor Networks (WSNs), one of fundamental tasks is to keep the nodes \emph{physically connected}
%and to reduce the energy consumption.
%%since \emph{logical connectivity} does not necessarily means the physical connectivity.
%%Here, a pair of nodes is logically connected means the network is connected in topology.
%Physically connectivity means any pair of neighboring nodes have rendezvous active time slots to communicate with each other.
%%So the ``physical connectivity" necessarily requires the ``logical connectivity".
%%For example, two neighboring nodes $u$ and $v$ are logically connected but not physically connected when their active (or sleeping) time is not cooperated.
%Previous works let nodes have rendezvous active time to communicate with each other in part time and to sleep in the rest time in order to save energy.
%But it cannot necessarily achieve the physical connectivity to have rendezvous active time slots because of asynchronization.
WSNs have been applied in various environments such as ecological surveillance~\cite{mo2009canopy}.
Because of hardware limitation, sensor nodes have limited energy and unprecise clocks.
Various approaches have been designed to save energy and improve some network performances on throughput, delay and per-node fairness.
%To save energy,  nodes are often turned off or in sleep mode regularly.
In order to achieve good cooperation among nodes, synchronization protocols, e.g FTSP~\cite{maróti2004flooding},
    were proposed but considerable energy and time were consumed especially in large scale networks as well.
How to design protocols to guarantee the communication among nodes under the asynchronous network becomes a very
    critical and challenging problem.
Media Access Control (MAC) protocols let nodes to know when and how to access common channels~\cite{ye2004medium}.
Some popular MAC protocols, such as TDMA- and CSMA-based, were designed to share communication
   medium among nodes by assigning each node some fixed active time slots in TDMA or by letting nodes locally contest their channel access in CSMA.
Both of two types of protocols try to build a physically connected network while controlling nodes' active time period in
   order to reduce energy consumption and improve network throughput.

TDMA has the advantage that time slots are previously scheduled to each node.
Therefore, a network can achieve high channel utilization under high media access contention and reduce collision among
   neighbors with a low cost when their clocks are well synchronized.
But TDMA also has some disadvantages~\cite{ye2004medium}, some of which are caused by clock asynchronization.
%In WSNs, sensor nodes usually have limited energy and their clock are easy to be strew or shift off.
%It will cost much extra energy and time in order to achieve the whole network synchronization especially in energy-limit4ed networks.
B-MAC~\cite{polastre2004versatile} adopts Low Power Listening (LPL) to solve communication failure caused by clock asynchronization.
Although CSMA doesn't austerely require the clock synchronization, it cannot achieve channel utilization as high as TDMA
  and costs additional time and energy on channel access contention.
Thus some hybrid MAC protocols, such as B-MAC~\cite{polastre2004versatile}, S-MAC~\cite{ye2004medium} and T-MAC~\cite{van2003adaptive}, combining the strengths of both TDMA and CSMA, were proposed.
These MAC protocols essentially adopt the LPL technique or improved LPL to allievate localized asynchronization problem.
However, they cannot avoid channel contention and obtain channel allocation fairness in many scenarios~\cite{jian2008can}.
Thus another challenging problem is to decrease the unfair contention for medium access without synchronization while increasing the channel allocation fairness.

This paper designs a localized scheme, named Adjustment of Duty Cycle (\protocol), based on \emph{Quorum Systems} (QS)~\cite{malkhi1998byzantine}, to adaptively adjust the duty cycle of each node.
A QS is a set of subsets of a universe set such that every pair of subsets intersect with no empty.
In recent years, QS is applied to establish control channels in dynamic spectrum access networks~\cite{bian2009quorum}, to save power~\cite{wu2008fully}, and to maximize throughput in
limited information multiparty MAC~\cite{chaporkar2006dynamic}.
By \protocol, each node can select sufficient amount of active time slots, composing a set $\subperiod\subset\period$
  (called a quorum) according to the amount of its demand while it can sleep to save energy at its rest time in a period.
Therefore, its duty cycle $|\subperiod|/|\period|$ is adaptively adjusted when the amount of active time slots
   $|\subperiod|$ is changed.
Each node will inform its neighbors of a quorum it selected thus the channel contention among them is decreased.
%\emph{The intersection property} of QS ensures that any pair of neighboring nodes have rendezvous active time to communicate with each other.
%And \emph{the rotation closure property} of QS guarantees that any pair of neighboring
%nodes can have rendezvous active time to communicate with each other without synchronization.
%In \protocol, we design a QS $\QS_g$ to satisfy the rotation closure property and an algorithm (Algorithm~\ref{algorithm:quorum selection}) to select quorums for every nodes.
The contributions of this paper are as follows:
\\$\bullet$ \protocol can adaptively adjust duty cycle by demand, and
  increase the channel allocation fairness comparing with existing contention-based MAC protocols.
\\$\bullet$ \protocol guarantees each pair of nodes having sufficient rendezvous active time to implement demand, and
      the worst case of channel utilization is lower bounded.
\\$\bullet$ By \emph{the rotation closure property} and \emph{intersection property} of QS, the successful connectivity of a whole network is guaranteed even under asynchronization so no extra energy is consumed on synchronization.
\\$\bullet$ This paper analyzes the performance of \protocol under data aggregation, and derives the impact of QS load
  on network delay, which is defined as a duration from one moment some data is sampled to another that all data is received by the sink.

The organization of this paper is as follows.
We first give the network model, formulate our problem and introduce the QS technology in Section~\ref{sec:system model and preliminary}.
In Section~\ref{sec:quorum system based time slot assignment}, we design our protocol \protocol and address its preliminary properties.
Under clock synchronization and asynchronization, the performances of \protocol are presented when certain demand is implemented in Section~\ref{sec:synchronous demand implementation} and~\ref{sec:asynchronous demand implementation} respectively.
In Section~\ref{sec:experiment}, we implement our scheme \protocol in a real test-bed consisting of TelosB nodes and
   evaluate its performance on real systems.
Section~\ref{sec:related work} tells the relative work in recent years.
The work of paper is concluded in Section~\ref{sec:conclusion}.
%%%%%%%%%%%%%%%%%%%%%%%
\section{System Model and Preliminary}
\label{sec:system model and preliminary}
\subsection{Network Model and Quorum System}\label{subsection:network model}
%In this paper, we consider two kinds of networks: \emph{random network} and \emph{arbitrary network}.
%In a random network, the node density keeps constant $\nodedensity$ while it does not in an arbitrary network.
We formulate a network by a graph $G(V,E)$, where $V$ (or $E$) is a set of all nodes (or edges).
Let $n$ denote the number of total nodes and $\sink$ denote an only sink in the network.
Each node is assigned a unique ID.
A radius of the network $G$ with respect to $\sink$, denote by $\graphr$, is defined as the maximum distance (hops) between $\sink$ and those nodes in $G$.
This paper studies the duty cycle adjustment under several popular interference models (denoted by $\interferenceM$):
RTS/CTS, protocol model and physical model~\cite{li2008capacity}.
%RTS/CTS ($\interferenceM_{rc}$), protocol model ($\interferenceM_{pr}$), physical model ($\interferenceM_{ph}$)~\cite{li2008capacity}.

QS, denoted by $\QS$, was used and introduced in precious papers~\cite{tseng2003power}\cite{jiang2005quorum}\cite{bian2009quorum}.
A QS $\QS\subset2^{\period}$, containing quorums, denoted by $\quorum$, is a set of subsets $\subperiod$ of $\period$, where $\period=\{\slot_1,\cdots,\slot_{\periodsize}\}$ is a period and composes of $\periodsize$ time slots.
%Let $\slotsetsize_u=|\subperiod_u|$  and $\periodsize=|\period|$.
A \emph{rotation} of a quorum $\quorum$ is defined as $\rotate$($\quorum$, $i$)=\{($\slot_j+i$) mod $\periodsize|\slot_j\in \quorum$\}, where $i$ is a non-negative integer.
Some QSs satisfy the \emph{rotation closure property}, \ie $\forall i\in\{0,\cdots, \periodsize-1\}: \quorum_1\cap \rotate(\quorum_2,i)\neq \varnothing$, where $\quorum_1,\quorum_2\in \QS$.
\begin{lemma}\label{lemma:rotation satisfication}
Grid, torus and cyclic QS all satisfy the \emph{rotation closure property}~\cite{jiang2005quorum}.
\end{lemma}
\subsection{Problem Statement}\label{subsect:problem statement}
Two neighboring nodes $u$ and $v$ can communicate with each other in WSNs \emph{iff}  they have at least one rendezvous active time slot.
%%\begin{equation}\label{equ:common active}
% \forall u,v\in V\ and\ v\in\communicateset_u:\subperiod_u\cap\subperiod_v\neq \varnothing, \subperiod_u,\subperiod_v\subset\period
%\end{equation}
When a network is asynchronous, \ie the clock of each node $u$ has a shift $\clockshift^u\geq 0$ from real time, the set of $u$'s active time slots $\subperiod_u$ accordingly has a shift, \ie $\subperiod_u'=\subperiod_u+\{\clockshift^u\}$.
The following equation should be satisfied if a pair of neighboring nodes can communicate with each other under asynchronization.
\begin{equation}\label{equ: asynchronous common active}
 \forall u,v\in V\ \mbox{and}\ v\in\communicateset_u:\subperiod_u'\cap\subperiod_v'\neq \varnothing, \subperiod_u',\subperiod_v'\subset\period
\end{equation}
where $\communicateset_u$ is a \emph{communication set} centered at a node $u$ and a set containing $u$ and those nodes in its communication range.
Here, we call a pair of nodes as neighboring when they respectively belong to the communication set of each other.
Equation~(\ref{equ: asynchronous common active}) means a pair of neighboring nodes can be physically connected even
  under asynchronization if their active time slot sets are properly designed.
%Notice the physical connectivity under the synchronization is a special case in Equation~(\ref{equ: asynchronous common active}).

%The network usually affords of different kinds of tasks, such as data aggregation and collection.
%And each node $u$ correspondingly has its own task to do, where the task of $u$ is the amount of data which $u$ should transmit or receive in unit time.
%Thus each node requires a certain amount of time to finish its tasks.
We define a parameter \emph{demand} $\demand$ to indicate the amount of data needed to transmit or receive in unit time.
Notice that Equation~(\ref{equ: asynchronous common active}) indicates a pair of neighboring nodes should have not only common active time but also enough time to finish all of its demand.
This paper aims to deriving a \emph{demand condition} so that each node $u$ can implement its demand by locally choosing a subset $\subperiod_u'\subset\period$ to guarantee each pair of nodes in $\communicateset_u$ satisfying Equation~(\ref{equ: asynchronous common active}).
To obtain the purpose, this paper designs the localized duty cycle adjustment scheme \protocol.
\section{Quorum System based Time Slot Assignment}
\label{sec:quorum system based time slot assignment}
This section presents our designing of \protocol and analyzes its properties.
These properties indicate \protocol can achieve better solution than existing protocols on fair medium
     access even under asynchronization as described in Section~\ref{section:introduction}.
\subsection{Designing of \protocol}\label{subsection:designin of qsts}
\protocol lets each node $u$ obtain a time slot set $\subperiod_u$ so that three following conditions can be satisfied:
\ding {172} Equation~(\ref{equ: asynchronous common active}); %This means that any node would not miss chance to communicate with others when the clock is either synchronous or asynchronous.
\ding {173} Demand condition; % \protocol can satisfy the demand of each node if the demand of each node satisfies the demand condition;
\ding {174} The active time is minimized to save energy.

\protocol consists of two steps.
The first step is to design quorums in a grid QS, denote by $\QS_g$.
In the second step, each node locally selects its quorum base on one-hop information about selected quorums.
\begin{figure}[h]\hspace{0.2cm}
\begin{minipage}[t]{0.45\linewidth}
\centering \includegraphics[scale=0.9, bb=235 371 357 464]{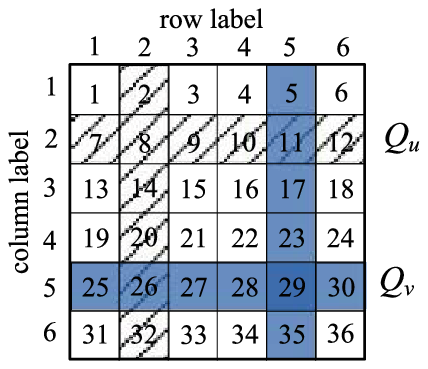}
\centering\caption{\footnotesize\label{fig:different_nodes_quorum} A grid QS $\QS_g$ contains $\period$. There are $\lceil\sqrt{\periodsize}\rceil$ rows and columns in $\QS_g$.}
\end{minipage}
\hspace{0.3cm}\begin{minipage}[t]{0.45\linewidth}
\centering \includegraphics[scale=0.9, bb=235 371 357 464]{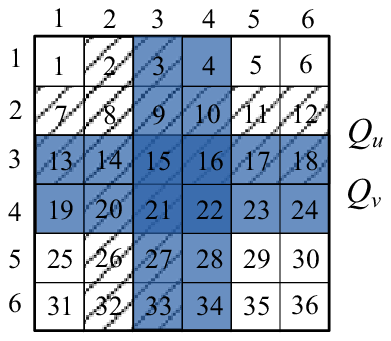}
\centering\caption{\footnotesize\label{fig:quorum overlap} $\quorum_u$ and $\quorum_v$ overlap at $3^{th}$ line and column.}
\end{minipage}
\end{figure}

At the first step, we construct a grid QS $\QS_g$ with the size
   $\lceil\sqrt{\periodsize}\rceil\times\lceil\sqrt{\periodsize}\rceil$
   based on a period $\period$ as shown in Figure~\ref{fig:different_nodes_quorum}.
The time slots from $\period$ are allocated into the grids of $\QS_g$ from right to left
   in a row-major manner as shown in Figure~\ref{fig:different_nodes_quorum}.
In each period, node $u$ requires a set of time slots $\subperiod_u$ to afford of its demand $\demand_u$.
%Denote the cardinality of $\subperiod_u$ by $\slotsetsize_u$.
The relation between the cardinality of $\subperiod_u$, denoted by $\slotsetsize_u$, and $\demand_u$
   is $\slotsetsize_u=\frac{\demand_u|\period|}{\datarate}$, where $\datarate$ is the data rate.
Thus we design a quorum $\quorum_u$ for $u$ and $\quorum_u$ contains
   $\lceil\frac{\demand_u\sqrt{\periodsize}}{2\datarate}\rceil$ rows and columns.

\emph{Rule 1 [Quorum Design]}:
The quorums $\quorum_i\in \QS_g$ are organized by from
   $i^{th}$ to $i+\lceil\frac{\demand_u\sqrt{\periodsize}}{2\datarate}\rceil^{th}$ rows and from $i^{th}$ to $i+\lceil\frac{\demand_u\sqrt{\periodsize}}{2\datarate}\rceil^{th}$ columns, where the row and column numbers are labeled as shown in Figure~\ref{fig:different_nodes_quorum} and $i\in\mathbb{Z}+$ and $i\leq\sqrt{\periodsize}+1-\lceil\frac{\demand_u\sqrt{\periodsize}}{2\datarate}\rceil$.

%(Notice that the detailed design of QS is needed. There are several questions:1. If the quorum contains $i^{th}$ row and $j^{th}$ column, then the (i,j) time slot is wasted. 2. A node $u$ can know the data which it tries to transmit, but do not know the data it will receive. If $u$ spends time to waiting for receiving data, and $u$ may waste time.)
%Each node $u$ should share the time with its neighbors thus every node has chance to access the medium and common time slots with its destination node.
%Firstly, we give a notion \emph{communication set} centered at a node $u$, denoted by $\communicateset_u$.
% It is a set containing $u$ and those nodes in the communication range of $u$.
At the second step, we design a quorum selection method for each node as described in Algorithm~\ref{algorithm:quorum selection}, in which ``a quorum $\quorum_i$ is occupied" means that $|\quorum_i\cap\quorum_j|>\lceil\frac{\demand_i\demand_jm}{2\datarate^2}\rceil$, $i\neq j$ when $\quorum_i$ is selected earlier than $\quorum_j$.
Figure~\ref{fig:quorum overlap} shows that $\quorum_u$ is occupied by $\quorum_v$ if  $\quorum_u$ is selected earlier than $\quorum_v$.
Notice that the parameter $\selecttime$ in Algorithm~\ref{algorithm:quorum selection} will be discussed in Lemma~\ref{lemma:load upper bound} and~\ref{lemma:load multi k single ball}.
Its value is determined in advance.
\begin{algorithm}
\caption{Quorum Selection}
\label{algorithm:quorum selection}
\textbf{Input}: $\QS_g$, which is allocated time slots from $\period$.\\
\textbf{Output}: Each node $v\in\communicateset_u$ is allocated a quorum.

\begin{algorithmic}[1]
\STATE $v\in\communicateset_u$ sets a positive natural $k=1$;
\STATE $v$ sets a list $L_o$ storing the quorums occupied by others.
\WHILE{$\communicateset_u\neq\emptyset$}
\STATE $v$ randomly selects a quorum $\quorum_i$ with equal probability, where $i=1,\cdots,\lceil\sqrt{\periodsize}+1\rceil-\lceil\frac{\demand_u\sqrt{\periodsize}}{2\datarate}\rceil$;
\IF{\label{quorum selection if}$\quorum_i$ was occupied and $k\leq\selecttime$}
\STATE Add $\quorum_i$ into $L_o$;
\STATE $v$ selects another quorum $\quorum_j$, $i\neq j$ and $\quorum_j\notin L_o$; $k\ +=1$;
\ENDIF
\STATE Set $k=1$ and delete $v$ from $\communicateset_u$.
\ENDWHILE
\STATE $v$ informs its neighbors of its quorum in a message.
\end{algorithmic}
\end{algorithm}

%%%%%%%%%%%%%%%%%%%%%%%%%%%%%%%%%%%%%%%%%%%%%%%%%%%%%%%
\subsection{Properties of \protocol}
The properties of \protocol includes \emph{physical connectivity}, \emph{maximum demand} and \emph{maximum load}.
We analyze the effect of \protocol on network connectivity when the clocks are either synchronous or asynchronous.
%%%%%%%%%%%%%%%%%%%%%%%%%%%%%%5
\\\textbf{The physical connectivity.}
Physical connectivity is the preliminary condition under which nodes can communicate with each other and implement their demand.
In this paper, physical connectivity  means the \emph{time-to-rendezvous},
%A node $v\in\communicateset_u$ has rendezvous active times $\slot_i$ with $u$, then $\slot_i$ is a \emph{rendezvous} slot of $u$ and $v$.
which is the amount of rendezvous slots between arbitrary nodes $u$ and $v$.
If a node $v\in\communicateset_u$ and $u$ are physically connected, the time-to-rendezvous between them should be at least one slot.
%However, when the time-to-rendezvous between two nodes increases, the physical connectivity between them is not necessarily better.
%Because the load of each time slot or each quorum may be higher.
%Physical connectivity impacts the chance to the medium access of MAC and the load, which reflects the extent that nodes compete with others for same time slots.
% %\ie the frequency in which the time slots are used.
%When a node has higher demand, its demand can be satisfied better if the network has high physical connectivity and each time slot has low load.

By \protocol, the physical connectivity under both clock synchronization and asynchronization are different.
When the clock is synchronous, we can easily obtain Lemma~\ref{lemma:synchronous QSTS conncectivity}.
For example, two quorums $\quorum_u$ and $\quorum_v$ respectively contain one row and one column in Figure~\ref{fig:different_nodes_quorum}.
They rendezvous at the time slot 11 and 26.
Thus they are physically connected and have two rendezvous active time slots.
If one node (for example $v$) chooses a used quorum $\quorum_u$ as the instance in Line~\ref{quorum selection if} of Algorithm~\ref{algorithm:quorum selection},
then two quorums (for example, $\quorum_v$ and $\quorum_u$) have more than two rendezvous time slots.
\begin{lemma}\label{lemma:synchronous QSTS conncectivity}
If two arbitrary nodes $u$ and $v\in\communicateset_u$  are respectively allocated quorums $\quorum_u$ and $\quorum_v$ according to Algorithm~\ref{algorithm:quorum selection},
they have at least $\lceil\frac{\demand_u\demand_vm}{2\datarate^2}\rceil$ rendezvous active time slots when clocks are synchronous.
\end{lemma}

When clocks are asynchronous, a node $u$ is prone to have clock shift $\clockshift^u$, which is the difference between the local time of $u$'s clock and the exact time.
Therefore, the relative clock shift between $u$ and $v$ is $\clockshift(u,v)=\clockshift^u-\clockshift^v$.
This paper always assumes $\clockshift^u, \clockshift(u,v)<+\infty$.
\begin{lemma}\label{lemma:asynchronous QSTS conncectivity}
Any pair of quorums in the same QS $\QS_g$ must have at least $\lceil\frac{\demand_u\demand_vm}{4\datarate^2}\rceil$ rendezvous active time slots
even when the relative clock shift between any pair of nodes is an arbitrary value.
%\ie any pair of nodes are physically connected when the clock asynchronization exists.
\end{lemma}
\begin{proof}
A grid QS $\QS_g$ satisfies the rotation closure property according to Lemma~\ref{lemma:rotation satisfication}.
Thus any two quorums $\quorum_u$ and $\quorum_v$ satisfy $\quorum_u\cap\rotate(\quorum_v,i)\neq \varnothing$, where $i=0,\cdots,\periodsize-1$, when $\quorum_u$, $\quorum_v\in \QS_g$.
For any pair of nodes $u$ and $v$, there is relative clock shift $\clockshift(u,v)$ because of the clock asynchronization.
Without loss of generality, let the rotation of $\quorum_v$ be $\rotate(\quorum_v,\clockshift(u,v))$.
It means that the slots in $\quorum_v$ shift because of the relative clock shift.
Thus $\rotate(\quorum_v,\clockshift(u,v))\cap\quorum_u\neq\varnothing$ for $i=0,\cdots,\periodsize-1$ if $v\in\communicateset_u$ since $\clockshift(u,v)\ mod\ \periodsize$ is positive and not bigger than $\periodsize-1$.
It means $u$ and $v$ are physically connected because a period totally contains $\periodsize$ time slots.

Next we look for a lower bound of the cardinality of $\rotate(\quorum_v,\clockshift(u,v))\cap\quorum_u$.
In \protocol, any two quorums in the same QS has at least $\lceil\frac{\demand_u\demand_vm}{4\datarate^2}\rceil$ rendezvous active time slots according to Theorem 3 of~\cite{bian2009quorum} and Lemma~\ref{lemma:synchronous QSTS conncectivity},
\ie $|\rotate(\quorum_v,\clockshift(u,v))\cap\quorum_u|\geq\lceil\frac{\demand_u\demand_vm}{4\datarate^2}\rceil$.
\end{proof}

Notice that Lemma~\ref{lemma:asynchronous QSTS conncectivity} is obtained without considering the interference since our scheme \protocol is applied to MAC.
Thus if $|\clockshift(u,v)|$ \emph{mod} $\periodsize$ =0 for any pair of nodes $u$ and $v$, then the physical connectivity of \protocol under the clock asynchronization is same with that under the clock synchronization, and the same result can be obtained on the parameter load.
In the subsection~\ref{subsection:constructing a tree}, we will analyze the effect of the interference.
%
%$\longleftrightarrow$
%
%The next paragraph should be inserted into the section ??
%Here we consider the interference model and use a constant $\coninterf$ to denote the interference range.
%We
%
%$\longleftrightarrow$

Notice that other kind of QSs are also applicable in \protocol according to Lemma~\ref{lemma:rotation satisfication} in spite that clocks are synchronous or asynchronous.
If a QS should satisfy the rotate closure property, the quorums in the QS should satisfy the condition in Lemma~\ref{lemma:quorum cardinality in closure property}.
\begin{lemma}\label{lemma:quorum cardinality in closure property}
If a QS $\QS$ satisfies the rotate closure property, then the cardinality of any quorum in $\QS$ must be more than $\sqrt{\periodsize}$~\cite{jiang2005quorum}.
\end{lemma}
%%%%%%%%%%%%%%%%%%%%%%%%%%%%%5
\textbf{The maximum demand.} Here the demand condition is given so each node $u$ can afford its demand $\demand_u$.
It is easy to know $\demand_u\leq\datarate$ if the demand $\demand_u$ is implementable.
When the interference models $\interferenceM$  presented in the subsection~\ref{subsection:network model} are considered, the demand $\demand_u$ cannot necessarily be close to $\datarate$.
The summation of all nodes in the same communication set $\communicateset_u$ should satisfy the condition in
Lemma~\ref{lemma:bound on demand} if all of the demand of nodes in $\communicateset_u$ can be implemented.
Before we give out Lemma~\ref{lemma:bound on demand}, we introduce a constant $\coninterf(\interferenceM)$ which is determined by the interference model $\interferenceM$.
We calculate  $\coninterf(\interferenceM)$ by the technique of vertex coloring.
The vertex coloring means to color all nodes with minimal number of colors under $\interferenceM$.
Thus nodes with same color are interference-free under $\interferenceM$.
\begin{lemma}\label{lemma:bound on demand}
When  demands of all nodes are implementable, the demands of nodes belonging to the same $\communicateset_u$ should satisfy $\sum\limits_{v\in\communicateset_u}\demand_v\leq \frac{\datarate}{\coninterf(\interferenceM)}$, where $\coninterf(\interferenceM)$ is a constant related with the interference model $\interferenceM$.
\end{lemma}
\begin{proof}
When no interference is involved, the nodes in $\communicateset_u$ share a period.
So $\sum\limits_{v\in\communicateset_u}\slotsetsize_v\leq T\Rightarrow
\sum\limits_{v\in\communicateset_u}\frac{\demand_v\period}{\datarate}\leq T\Rightarrow
\sum\limits_{v\in\communicateset_u}\demand_v\leq \datarate$.
Under the interference model $\interferenceM$, each communicate set can transmit or receive in every $\coninterf(\interferenceM)$ periods in order to be interference-free.
The average data rate is $\frac{\datarate}{\coninterf(\interferenceM)}$.
So $\sum\limits_{v\in\communicateset_u}\demand_v\leq \frac{\datarate}{\coninterf(\interferenceM)}$.
\end{proof}

We find it is not always suitable to decrease the maximum demand of each node since a node should keep active in at least $\sqrt{\periodsize}$ time slot to satisfy the rotate closure property according to Lemma~\ref{lemma:quorum cardinality in closure property}.
Hence demand of a node $u$ should have a lower bound.
Since the $\quorum_u$ contains $\lceil\frac{\demand_u\sqrt{\periodsize}}{2\datarate}\rceil$ rows and a column, the lower bound can be obtained by Lemma~\ref{lemma:demand lower bound}.
\begin{lemma}\label{lemma:demand lower bound}
When the $\QS_g$ satisfies the rotation closure property, the demand of each node should satisfy $\demand_u\geq \conone\datarate$, where $\conone=1-\sqrt{1-1/\sqrt{\periodsize}}$, when $\periodsize>1$.
\end{lemma}
\begin{proof}
For an arbitrary node $u$, the cardinality of $\quorum_u$ is $2\sqrt{\periodsize}\lceil\frac{\demand_u\sqrt{\periodsize}}{2\datarate}\rceil-\lceil\frac{\demand_u\sqrt{\periodsize}}{2\datarate}\rceil\times\lceil\frac{\demand_u\sqrt{\periodsize}}{2\datarate}\rceil$.
By Lemma~\ref{lemma:quorum cardinality in closure property}, we have $2\sqrt{\periodsize}\lceil\frac{\demand_u\sqrt{\periodsize}}{2\datarate}\rceil-\lceil\frac{\demand_u\sqrt{\periodsize}}{2\datarate}\rceil\times\lceil\frac{\demand_u\sqrt{\periodsize}}{2\datarate}\rceil\geq\sqrt{\periodsize}$.
$\Rightarrow$$2\lceil\frac{\demand_u\sqrt{\periodsize}}{2\datarate}\rceil-\lceil\frac{\demand_u}{2\datarate}\rceil\times\lceil\frac{\demand_u\sqrt{\periodsize}}{2\datarate}\rceil\geq 1$.
When $\periodsize=1$, the above inequality can be always satisfied because $\demand_u/\datarate\leq 1$.
When $\periodsize>1$, $\demand_u\geq \conone\datarate$, where $\conone=1-\sqrt{1-1/\sqrt{\periodsize}}$.
\end{proof}

Let $D_k=\min\limits_{v\in\communicateset_u}\demand_v$ so $D_k\geq \conone\datarate$ according to Lemma~\ref{lemma:demand lower bound}.
Thus we have $\sum\limits_{v\in\communicateset_u}\demand_v\geq \conone|\communicateset_u|\datarate$.
%In Lemma~\ref{lemma:demand lower bound}, the node would have no task to do in some active time slots if the demand of the
%    node is less than $\conone\datarate$ (when $\periodsize>1$).
%When the demand of each node is very low, \ie $\demand_u/\datarate<\contwo$, where $\contwo$ is a small positive constant and may be less than 2\%, the duty cycle can be very low.
%Then existing MAC protocols can be applied, such as B-MAC~\cite{polastre2004versatile} and SCP~\cite{ye2006ultra}.
%%%%%%%%%%%%%%%%%%%%%%%
\\\textbf{The load bound.}
When the condition in Lemma~\ref{lemma:bound on demand} is satisfied, there still exists competition between a pair of nodes $u$ and $v$ in the same $\communicateset$,
where the competition between them because $|\quorum_u\cap\quorum_v|>\lceil\frac{\demand_u\demand_vm}{2\datarate^2}\rceil$, $u\neq v$, as described in Algorithm~\ref{algorithm:quorum selection}.
Notice that the parameter $\selecttime$ in Algorithm~\ref{algorithm:quorum selection} is used to decrease the competition, where $\selecttime$ is the number of times in which the same node selects different quorums.
It is important to take full advantage of the time diversity of medium access, \ie to decrease the competition between different quorums, in order to control the channel congestion.
Some previous work designed protocol to minimize the load~\cite{bian2009quorum}.
But it is not suitable under a more practical case in this paper.
That is each node has  demand different from others' because of their different network tasks.
Furthermore, we show that it cannot fully use the time slots when the demand is very low, under which it will degrade the channel utilization to decrease the load.
In \protocol, we present the upper- and lower-bound of the load under the demand constraint given in Lemma~\ref{lemma:bound on demand} and \ref{lemma:demand lower bound}.

We first give out a lower bound of the quorum load.
Two propositions 4.1 and 4.2 in~\cite{naor1994load} gave a result that $\load(\QS)\geq\max\{\frac{1}{\QSsizemin(\QS)},\frac{\QSsizemin(\QS)}{\periodsize}\}$, where $\QSsizemin(\QS)$ is the size of smallest quorum in $\QS$, and $\QSsizemin(\QS)\geq\sqrt{\periodsize}$ according to Lemma~\ref{lemma:demand lower bound}.
We then have $\frac{1}{\QSsizemin(\QS)}\leq\frac{1}{\sqrt{\periodsize}}$ and $\frac{\QSsizemin(\QS)}{\periodsize}\geq\frac{1}{\sqrt{\periodsize}}$.
Therefore, $\load(\QS)\geq\frac{\QSsizemin(\QS)}{\periodsize}$, which is different from the result in Theorem in~\cite{bian2009quorum}, because the rotation closure property is considered.

Next we discuss $\QSsizemin(\QS)$.
Notice that  the cardinality of $\quorum_u$ should be not less than $\sqrt{\periodsize}$ when the demand of each node $\demand_u\geq \conone\datarate$.
Otherwise, the rotation closure property cannot be satisfied.
Thereinafter, we analyze the bound of the load in two case: $\selecttime=1$ and $\selecttime\geq 2$.

In Lemma~\ref{lemma:bound on demand},
an obvious upper bound of the QS load appears when a node is required to afford the full demand.
\ie, the demand $\demand_v$ of a node $v\in\communicateset_u$ is not less than the maximal data rate, $\demand_v\geq\datarate$.
Under this case, the cardinality of each  quorum is $\periodsize$.
Because each node randomly selects a quorum with equal probability and the probability that each time slot is included in a quorum is $\frac{1}{|\communicateset|}$, $\load_{\strategy}(i)=\sum\limits_{Q\in\QS:\slot_i\in Q}P_{\strategy}(Q)=\sum\limits_{Q\in\QS:\slot_i\in Q}\frac{1}{|\communicateset|}=1$, $\load_{\strategy}(\QS)=1$.
Under \protocol, the cardinality of each quorum is not bigger than $\periodsize$ and the probability that each time slot is included in a quorum is less than $\frac{1}{|\communicateset|}$. So $\load_{\strategy}(\QS)\leq 1$.
\begin{lemma}\label{lemma:load upper bound}
In \protocol, the load of $\QS_g$ is less than $\frac{\confour}{|\communicateset_u|}$, \ie, $\load_{\strategy}(\QS_g)\leq\frac{\confour}{|\communicateset_u|}$, when the QS $\QS_g$ satisfies the rotation closure properties and $\selecttime=1$, where $\confour=\lceil\frac{1}{\coninterf}\rceil-\lceil\frac{1}{4\coninterf^2}\rceil$.
\end{lemma}
\begin{proof}
Since the cardinality of a quorum $\quorum_v$ ($v\in\communicateset_u$) is $|\quorum_v|$ under \protocol, the probability that a time slot is included in $\quorum_v$ is $\frac{|\quorum_v|}{\periodsize}$.
Suppose there are $\gamma$ quorums in each $\QS$.
$|\communicateset_u|\leq \gamma\leq\sqrt{\periodsize}$ when the summation of all nodes' demand in the same $\communicateset_u$ is implementable.
So the probability that a quorum is chosen by a node is $\frac{1}{\sqrt{\periodsize}}\leq\frac{1}{\gamma}\leq\frac{1}{|\communicateset_u|}$.
Therefore, the probability that each time slot is included in the quorum $\quorum_v$ when there are $|\communicateset_u|$ quorums is $\frac{1}{|\communicateset_u|}\times\frac{|\quorum_u|}{\periodsize}$.
The load induced by the strategy $\strategy$ on a time slot $\slot_i$ is $\load_{\strategy}(\slot_i)=\sum\limits_{\quorum_u\in\QS_g:\slot_i\in \quorum_u}\frac{1}{|\communicateset_u|}\times\frac{|\quorum_u|}{\periodsize}$.
Notice there are $|\communicateset_u|$ nodes to select quorums from $\QS_g$, \ie $\QS_g=\{\quorum_v:\quorum_v\in\communicateset_u\}$. So
\begin{align}\label{equ:load}
%\begin{split}\small
&\load_{\strategy}(\slot_i)=\frac{1}{\gamma\periodsize}\sum\limits_{\quorum_u\in\QS_g:\slot_i\in \quorum_u}|\quorum_u|\nonumber\\
\leq&\frac{1}{|\communicateset_u|\periodsize}\sum\limits_{\slot_i\in\quorum_u\in\QS_g}\{2\sqrt{\periodsize}\lceil\frac{\demand_u\sqrt{\periodsize}}{2\datarate}\rceil-\lceil\frac{\demand_u\sqrt{\periodsize}}{2\datarate}\rceil^2\}\nonumber\\
=&\frac{1}{|\communicateset_u|}\sum\limits_{\slot_i\in\quorum_u\in\QS_g}\{\lceil\frac{\demand_u}{\datarate}\rceil-\lceil\frac{\demand_u}{2\datarate}\rceil^2\}
\leq\frac{1}{|\communicateset_u|}(\lceil\frac{1}{\coninterf}\rceil-\lceil\frac{1}{4\coninterf^2}\rceil)
%\end{split}
\end{align}

Thus the load induced by the strategy $\strategy$ on the quorum system $\QS_g$ is $\load_{\strategy}(\QS_g)=\max\limits_{\slot_i\in \period}\load_{\strategy}(\slot_i)\leq\frac{1}{|\communicateset_u|}(\lceil\frac{1}{\coninterf}\rceil-\lceil\frac{1}{4\coninterf^2}\rceil)
=\frac{\confour}{|\communicateset_u|}$, where $\confour=\lceil\frac{1}{\coninterf}\rceil-\lceil\frac{1}{4\coninterf^2}\rceil$.
\end{proof}

Now we analyze the upper bound of the QS load when $\selecttime\geq 2$.
%we give a definition of the $k$-Round Ball Placement ($\selecttime$-RBP) problem~\cite{li2009multiple}.
%\begin{definition}[$\selecttime$-RBP]
%In a $\selecttime$ round ball placement, balls are randomly placed into $\periodsize$ bins in $\selecttime$ rounds.
%Let $\periodsize_0=\periodsize$ be the original number of empty bins and $\periodsize_i$ be the number of empty bins after $i$ rounds.
%In the $i^{th}$ round, $b_i$ balls are randomly placed into $\periodsize_{i-1}$ remaining empty balls.
%$B=\sum^k_{i=1}b_i$ is the total number of balls placed.
%\end{definition}
Here we treat a quorum containing a row and a column and call the quorum as a bin, so each node actually selects several such bins according to the \emph{Rule 1}.
Each node has $\selecttime$ times to select its quorum in Algorithm~\ref{algorithm:quorum selection}.
Thus the quorum selection problem in Algorithm~\ref{algorithm:quorum selection} is equivalent to the \emph{k}-round ball placement problem~\cite{li2009multiple}.
When $m_i=1$, we can obtain that the maximum load achieved by Algorithm~\ref{algorithm:quorum selection} is less than $\frac{\log\log\sqrt{\periodsize}}{\log\selecttime}$ \whp according to Theorem 6 of~\cite{mitzenmacher2001power}.
When $m_i\geq 2$ in each round, \ie several bins are selected together in each round, it is equivalent to combining several bins into one.
Therefore, the total number of bins is correspondingly reduced.
We can obtain Lemma \ref{lemma:load multi k single ball}.
\begin{lemma}
\label{lemma:load multi k single ball}
The maximum load achieved by Algorithm~\ref{algorithm:quorum selection} is less than $\frac{\log\log\sqrt{\periodsize}}{\log\selecttime}$ \whp.
\end{lemma}
\section{Synchronous Demand Implementation}
\label{sec:synchronous demand implementation}
This section evaluates the performance of \protocol when the data aggregation is implemented and clocks are
    synchronous, which we call as synchronous demand implementation.
In order to implement the demands, we construct a tree and design the specific demand implementation methods.
The performance of \protocol under the asynchronization will be analyzed in the next
   section by applying the results of this section.
\subsection{Tree Construction}\label{subsection:constructing a tree}
Firstly, a tree~$\tree$ is constructed based on $G$ by constructing a connected dominator set (CDS).
we then define a new conception \emph{region} in order to obtain a conflict-free quorum assignment for each node.

We construct a CDS by the breadth-first-search (BFS) based on $G$.
Each dominatee connects with the dominator closest to it.
In this way, $\tree$, rooting at $\sink$, can be constructed and is ranked into $\levelmax$ levels from $\sink$.
The level of $\sink$ is labeled $\level_0$.
The parent and the children of a node $u$ are denoted by $p(u)$ and $c(u)$ respectively.

%For the data aggregation, it need a data aggregation tree.
%Since we construct a tree~$\tree$ rooted at the sink $\sink$ in the subsection~\ref{subsection:network model}.
%We directly use $\tree$ to implement the network demands:data aggregation and collection.
We assign QSs for each region in two phases.
At the first phase, a new conception \emph{region} is defined to obtain the conflict-free partition by the vertex coloring.
At the second phase, we assign each region with a period so nodes can be active without confliction.\\
\textbf{Phase I:}
After $\tree$ is constructed, each node can know its own level and its one-hop neighbors' IDs and levels.
The one-hop neighborhood of each node $u$ is denoted by $N_1(u)$ and notice that $u\in N_1(u)$.
We call a one-hop neighborhood of a non-dominatee node as a \emph{region} (denoted by $\region$) in the tree~$\tree$.
Notice that any dominatee does not form a region.
Because of interference within a network, the QSs of some neighboring regions cannot be assigned a same time slots set.
Here, we say two regions $\region_1$ and $\region_2$ are \emph{neighboring} (or \emph{over-lap}) if there are two nodes $u\in\region_1$ and $v\in\region_2$ and $u$ (or $v$) locates in the interference range of $v$ (or $u$).
\begin{figure}[h]
\centering
\subfigure[]{\includegraphics[scale=.7, bb=248 394 348 444]{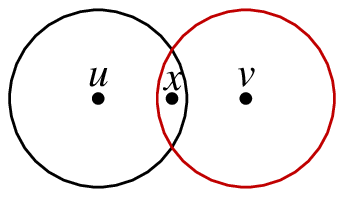}\label{subfig:overlap region}}
\hspace{1cm}\subfigure[]{\includegraphics[scale=.7, bb=235 387 357 459]{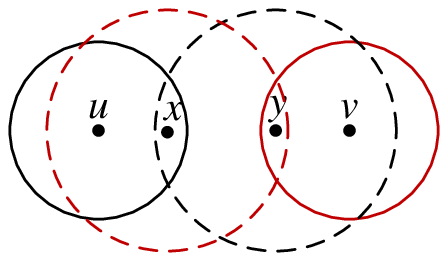}\label{subfig:neighboring region}}
\centering\caption{\footnotesize\label{fig:region} \subref{subfig:overlap region} $\region_u$ and $\region_v$ are two over-lap regions. \subref{subfig:neighboring region} $\region_u$ and $\region_v$ are two neighboring regions.}
\end{figure}

%The network composes of some regions.
If two regions are conflict-free, \ie they are not neighboring and over-lap, we color them with a same color.
So it is vertex coloring problem to find the minimal number of colors.
The least number of colors, denoted by $\colorc$, necessary to color all regions, is affected by the interference model
    $\interferenceM$, \ie we need at least $\colorc(\interferenceM)$ to color all regions and the regions with same color are conflict-free.
We label each region $\region_u$ with a color index $\colorindex_{\region_u}$, $\colorindex_{\region_u}\in\colorset=\{1,\cdots,\colorc(\interferenceM)\}$.
In order to determine the least number of colors, we define a parameter (denoted by $\interfpara$) to denote the interference range under different interference models.
%Then we have the following lemma.
%\begin{lemma}\label{lemma:interference parameter}
%The number $\interfpara$ of colors needed to color all regions in the whole network  is 85 under the interference models: $\interferenceM_{pr}$, $\interferenceM_{rc}$ and 13 under $\interferenceM_{ph}$.
%\end{lemma}
Thus any pair of regions can have same color if they are more than $\interfhop$ hops apart.
%, where $\interfhop(\interferenceM_{pr}$)=$\interfhop(\interferenceM_{ph}$)=1 and $\interfhop(\interferenceM_{rc})$=2.
\\\textbf{Phase II:} Allocate each region with a QS.
Because the quorum is designed in the subsection~\ref{subsection:designin of qsts}, each node belongs to at least one region and some belong to several regions.
For example, $x$ belongs to two regions $\region_u$ and $\region_v$ in Figure~\ref{subfig:overlap region}.
Thus the number of quorums each node occupies is same with the number of regions it belongs to.
We assign each region with a period,\ie a QS, according to the color so each neighboring or over-lap regions can be conflict-free.
For example, $\region_u$ and $\region_v$ are respectively assigned two periods $\period_u$ and $\period_v$.
Let $\period_u=\{1,2,3\}$ and $\period_v=\{4,5,6\}$ so $\period_u\cap\period_v=\varnothing$.

We use a natural number $i (i\in\mathbb{Z})$ to label the ID of a region in order to assign the time slot set conveniently.
Firstly, we color all the regions by Algorithm~\ref{algorithm:region coloring}.
Secondly, each region $i$ is assigned a slot set $\slotset_i$ by Algorithm~\ref{algorithm:slot set assignment}.
\begin{algorithm}
\caption{Region Coloring}
\label{algorithm:region coloring}
\textbf{Input}: The labels of all regions, $\region_i$, $i=1,\cdots,|CDS|$ and the color set $\colorset$.\\
\textbf{Output}: The colored regions.

\begin{algorithmic}[1]
\FOR{$i=1,\cdots,|CDS|$}
\IF{There is no region colored within $\interfhop$ hops centered at a region $\region_i$}
\STATE $\region_i$ labels itself with a color $\colorindex_{\region_i}$ ($\in\colorset$).
\ELSE
\STATE $\region_i$ selects a color $\colorindex_{\region_i}$ from $\colorset$ and $\colorindex_{\region_i}$ is different from the colors of other nodes within $\interfhop$ hops centered at $\region_i$'s.
\ENDIF
\ENDFOR
\end{algorithmic}
\end{algorithm}
\begin{algorithm}
\caption{Slot Set Assignment}
\label{algorithm:slot set assignment}
\textbf{Input}: All the colored regions, $\colorindex_{\region_i}$, $i=1,\cdots,|CDS|$ and $\colorindex_{\region_i}\in CDS$.\\
\textbf{Output}: Each region $\region_i$ obtains a slot set $\slotset_{\region_i}$.

\begin{algorithmic}[1]
\FOR{$j=\levelmax,\cdots,0$ and $\region_i\subset\level_j$}
\FOR{$i=1,\cdots,|CDS|$}
\STATE\label{line algorithm 1 in slot set} A region $\region^j_i$ colored with a color index $\colorindex_{\region^j_i}$ is assigned the slot set $\slotset_{\region^j_i}=\colorindex_{\region^j_i}\ mod\ \interfpara\times\slotsetsize$.
\IF{$\slotset_{\region^j_i}>\max\limits_{all\ c(i)}\slotset_{\region^{(j+1)}_{c(i)}}$}
\STATE $\rotate(\slotset_{\region^j_i}, m)$. %$\slotset_{\region^j_i} +=\period$.
\ENDIF
\ENDFOR
\ENDFOR
\end{algorithmic}
\end{algorithm}

After each region is assigned a period, it can obtain its quorum according to Algorithm~\ref{algorithm:quorum selection}.
We can design a determinate quorum selection method rather than the random one in Algorithm~\ref{algorithm:quorum selection}.
The determinate quorum selection method is given in Algorithm~\ref{algorithm:determinate quorum selection}.
The difference between Algorithm~\ref{algorithm:quorum selection} and Algorithm~\ref{algorithm:determinate quorum selection} is that each node doesn't select quorums randomly.
By Algorithm~\ref{algorithm:determinate quorum selection}, we can obtain some properties.

Here we give a notion: \emph{logical connection}.
Nodes $u$ and $v$ are logically connected if they locate in each other's transmission range.
If $u$ and $v$ cannot be active in common time slots, they cannot communicate with each other.
A graph can be logically connected by topology control algorithms.
%Algorithm~\ref{algorithm:region coloring} and \ref{algorithm:slot set assignment} do not change the logical connectivity of $\tree$.
Thus $\tree$ is logically connected if $G$ is.
According to Algorithm~\ref{algorithm:determinate quorum selection}, it is easy to know that each pair of nodes are physically connected in $\tree$ if $G$ is logically connected.
\begin{algorithm}
\caption{Determinate Quorum Selection}
\label{algorithm:determinate quorum selection}
\begin{algorithmic}[1]
\STATE Each node $u$ collects the level label and quorums of its neighbors in $N_1(u)$;
\STATE $u$ classifies its neighbors into three sets: $S_1$, $S_2$ and $S_3$.  $S_1$ contains the node in $\level_{k-1}$ and $S_2$ contains the nodes in $\level_k$ and $S_3$ contains the nodes in $\level_{k+1}$;
\IF{$u$ is the sink}
\STATE $S_1=\emptyset$;
\ENDIF
\FOR{$i=1,\cdots,w$}
\IF{$\quorum_i$ is not chose}
\STATE According to the order $S_3$, $S_2$ and $S_1$, each node in $S_3$, $S_2$ and $S_1$ chooses $\quorum_i$.
\ENDIF
\ENDFOR
\end{algorithmic}
\end{algorithm}

\begin{lemma}\label{lemma:OHIQA conflict free}
Each pair of neighboring or over-lap regions, and each pair of links in a same region are conflict-free according to Algorithm~\ref{algorithm:determinate quorum selection}.
\end{lemma}
\begin{proof}
According to Algorithm~\ref{algorithm:region coloring}, any neighboring (or overlap) regions are colored different colors.
In Algorithm~\ref{algorithm:slot set assignment}, the regions $\region_i$ and $\region_j$ colored different colors are assigned different slot set $\slotset_{\region_i}$ and $\slotset_{\region_j}$.
$\slotset_{\region_i}\cap\slotset_{\region_j}=\varnothing$ according the line~\ref{line algorithm 1 in slot set} in Algorithm~\ref{algorithm:slot set assignment}.
In Algorithm~\ref{algorithm:determinate quorum selection}, each node is assigned a quorum different from that of others in the same region.
According to the definition of QS, there are only two nodes to be active simultaneously.
\end{proof}
\begin{lemma}\label{lemma:slot first or later}
 If a node $u$ (expect $\sink$) and its parent $p(u)$ are respectively assigned the time slots $\slot_u$ and $\slot_{p(u)}$, then $\slot_u<\slot_{p(u)}$ according to Algorithm~\ref{algorithm:determinate quorum selection}.
\end{lemma}

In Algorithm~\ref{algorithm:slot set assignment}, the active period of each node is earlier than that of its parent.
According to Algorithm~\ref{algorithm:determinate quorum selection}, the node in the level $\level_i$ is active earlier than that in $\level_{i-1}$ in the same region.
Thus we obtain Lemma~\ref{lemma:slot first or later}.
It is easy to obtain that each parent will transmit after it receives all packets from all of its children in the same period.
%%%%%%%%%%%%%%%%%%%%%%%%%%%%%%%%%%%%%%%%
\subsection{Data Aggregation}\label{subsection:data aggregation demand}
This section discusses performance of \protocol under data aggregation.
%Under data aggregation, children send their packets to their parent and the parent aggregates their packets into one packet.
%The process is repeated till the sink receives the aggregated packet.
Under data aggregation, the demand of a parent is the summation of its children, \ie $\demand_{p(u)}=\sum\limits_{u\in \region_{p(u)}}\demand_u$.
We can obtain that the maximal size of a QS among all QSs is determined by the maximal degree of $\tree$  so we have Lemma~\ref{lemma:maxsize for a QS}.
\begin{lemma}\label{lemma:maxsize for a QS}
$\max|\QS|=\max\limits_{u\in\tree}\degree^2_u$.
\end{lemma}

The delay of data aggregation is given in Theorem~\ref{theorem: the delay boundary}.
\begin{theorem}\label{theorem: the delay boundary}
The maximal delay of data aggregation is $O(\interfpara\periodsize\graphr+\degreemax^2)$ by Algorithm~\ref{algorithm:region coloring}
and \ref{algorithm:determinate quorum selection}.
\end{theorem}
\begin{proof}
%In any region $\region_i$, two arbitrary node $u$ and $v$, $u,v\in\region_i$.
%According to Lemma~\ref{lemma:one-hop connectivity by QS}, $u$ and $v$ have two common active time slots.
Each region $\region_i$ is assigned a QS $\QS_i$ so it costs $|\QS_i|$ to finish all the transmission in $\region_i$.
Since $|\QS_i|\leq\max|\QS|$ according to Lemma~\ref{lemma:maxsize for a QS}, $\max|\QS|=\max\limits_{u\in\tree} \degree^2_u=\degreemax^2$.
There exists an assignment method such that the number of different color is at most $\interfpara$ in each level of the tree $\tree$.
Let each region $\region_i$ be assigned a slot set $\slotset_{\region_i}$.
Thus it needs at most $\interfpara\slotset_{\region_i}$ in each level.
Because there are totally $\levelmax$ levels and the lower level $\level_i$ is active earlier than the higher level $\level_j$ ($i>j$), the sink costs
$\interfpara\slotset_{\region_i}\levelmax$ time to collect all data.
Because $\slotset_{\region_i}\leq \periodsize$, $\levelmax\leq\graphr$ and $\slotset_{\region_i}\leq\periodsize$, $\interfpara\slotset_{\region_i}\levelmax\leq\interfpara\periodsize\graphr$, the maximal delay of the data aggregation by our method is $\interfpara\periodsize\graphr+\degreemax^2$.
\end{proof}

Any schedule has delay at least $\graphr$ (or $D$), where $D$ is the radius of a network~\cite{li2009Efficient}.
When $\sink$ locates at the  center of a topology, the delay lower boundary can be reduced to be $(\interfpara\periodsize+1)D+\degreemax^2$, where $D=\graphr/2$~\cite{li2009Efficient}.
\section{Asynchronous Demand Implementation}
\label{sec:asynchronous demand implementation}
\protocol does not require the global clock synchronization.
This section aims to analyze the delay of data aggregation under asynchronization.
Existing algorithms are designed to bound the delay, such as time slot assignment algorithm in~\cite{WanHWWJ09}.
However, neighboring nodes may not be physically connected under asynchronization.
Thus an additional method is given to ensure each pair of neighboring nodes are physically connected in the subsection~\ref{subsection:quorum share}.
\subsection{Asynchronous Delay of Data Aggregation}
\label{subsection:Asynchronous Delay of Data Aggregation}
%We denote the max clock shift by $\clockshiftMax$.
%And we suppose it is arbitrary and bounded, \ie $\clockshiftMax<\infty$.
%A clock shift of a node $u$ is denoted by $\clockshift^u(\realtime)$, which is the difference between the real time $\realtime$ and the local time $t^u$ of $u$, \ie $\clockshift(\realtime)=t^u-\realtime$.
We assume that the clock shift $\clockshift^u(\realtime)$ of a node $u$ randomly and uniformly distributes in the interval $[-\infty,\infty]$, where $t^u$ and $t$ are local time and exact time.
When $u$ selects a quorum $\quorum_u$, $u$ is actually active in $\rotate(\quorum_u, \clockshift(t))$ because of the clock shift  $\clockshift$.
According to Algorithm~\ref{algorithm:slot set assignment}, every region is assigned a period in every $\interfpara$ periods.
Without loss of generality, $u\in\region_u$ is active in a period $\period_{u+i\interfpara}$, where $i=1,2,\cdots$.
For a pair of neighboring nodes $u$ and $v$, the relative clock shift between them is $\clockshift^{uv}(t)=t^u-t^v$.
Thus $u$ and $v$ have common active time slots only if $\rotate(\quorum_v, \clockshift^{uv}(t))\cap\quorum_u\neq\varnothing$, where $\quorum_u\in\period_{u+i\interfpara}$.
We have the following lemma.
\begin{lemma}\label{lemma:asynchronous physical connectivity}
If $u$ and $v$ locate in the same region and $i-1<\frac{\clockshift^{uv}(t)}{\period\interfpara}<i+1$, where $i=1,2,\cdots$, $u$ and $v$ are physically connected.
\end{lemma}

According to Lemma~\ref{lemma:asynchronous physical connectivity}, nodes may not be physically connected because the data aggregation scheme is adopted under the clock asynchronization.
Notice that it is not caused by \protocol.

Asynchronous clock causes additional delay on the delay of data aggregation in order to ensure each pair of $u$ and $v$ could have common active time slots to communicate with each other in $\period_u$ under asynchronization.
That means $u$ and $v$ have to postpone their communication because of the clock shift $\clockshift^{uv}$.
But they can communicate with each other within at most $\interfpara-1$ additional periods delay if they are physically connected according to Lemma~\ref{lemma:asynchronous physical connectivity}.
Therefore, we can obtain the delay of data aggregation based on Theorem~\ref{theorem: the delay boundary} as illustrated in the following lemma.
\begin{lemma}
The delay of data aggregation is $O((2\interfpara-1)\period\graphr+\degreemax^2)$ under asynchronization if each pair of neighboring nodes are physically connected.
\end{lemma}
%
%Although QS-based time slot assignment methods do not necessarily guarantee the physical connectivity when the clock asynchronization exists.
%We can present the lower bound of probability that any pair of neighboring nodes are physically connected by our scheme \protocol.
%According to Lemma~\ref{lemma:asynchronous physical connectivity}, a pair of nodes $u$ and $v$ can be physically connected with the probability $\frac{3}{\interfpara}$.
%Thus a pair of nodes $u$ and $v$ cannot be physically connected with the probability $1-\frac{3}{\interfpara}$.
%We introduce a variable $X_{uv}$ to denote whether the pair of nodes $u$ and $v$ are physically connected or not.
%If $u$ and $v$ are physically connected, $X_{uv}=1$
%In $\tree$, there are totally $n-1$ pairs of nodes.
%The probability that any pair of nodes are physically connected seems very little as $n$ increases.
%(ignoreing***************)

%\begin{theorem}
%\end{theorem}
%\subsection{Asynchronous Delay of Data Collection(Analysis for data collection is unfinished)}
%\label{subsection:Asynchronous Delay of Data collection}
\subsection{Quorum Share}\label{subsection:quorum share}
Although Lemma~\ref{lemma:asynchronous QSTS conncectivity} ensures each pair of nodes (including non-neighboring) are physically connected,
Lemma~\ref{lemma:asynchronous physical connectivity} indicates each pair of neighboring nodes must be unable to communicate with each other in some periods since Algorithm~\ref{algorithm:slot set assignment} assigned each region with discontinuous periods.
This section designs a scheme to solve this problem.
%We find that any pair of nodes are actually physically connected under \protocol according to Lemma~\ref{lemma:asynchronous QSTS conncectivity}.
%But each region is assigned a period in Algorithm~\ref{algorithm:slot set assignment} and these periods are not continues.
Suppose $\region_u$ is assigned periods $\period_i$, $i=0,\interfpara, 2\interfpara, \cdots$, and the clock shift of a node $v\in\region_u$ is $\clockshift^v$ and $\quorum_v$ contains the time slots set $\subperiod_v$.
When $\subperiod'_v=\subperiod_v+\{\clockshift^v\}$ locates in the periods which does not satisfy the inequality in Lemma~\ref{lemma:asynchronous physical connectivity}, $v$ conflicts with some nodes, \ie the quorum $\quorum_v$ shifts into the QS of some other nodes $x_i$, $i=1,2,
\cdots, \leq |\communicateset_x|$, in another region $\region_x$, where $v$ locates in the interference range of $x_i$.
Notice that there are two kinds of quorum shifting: (1) a quorum $\quorum_v$ only shifts between periods when $\clockshift\ mod\ |\period|=0$; (2) a quorum $\quorum_v$ shifts among a QS when $\clockshift\ mod\ |\period|\neq 0$.

In our scheme described in Algorithm~\ref{algorithm:quorum share},  $v$ and $x_i$, $i=1,2, \cdots, \leq |\communicateset_x|$ share the quorum $\quorum_v$ with equal probability.
Our scheme can deal with two kinds of quorum shifting.
\begin{algorithm}
\caption{Quorum Share}
\label{algorithm:quorum share}
%\textbf{Input}: All the colored regions, $\colorindex_{\region_i}$, $i=1,\cdots,|CDS|$ and $\colorindex_{\region_i}\in CDS$.\\
%\textbf{Output}: Each region $\region_i$ obtains a slot set $\slotset_{\region_i}$.\\
\begin{algorithmic}[1]
\STATE Each node $v$ sets a list $L_t$ storing the nodes' ID, which occupy $v$'s quorum $\quorum_v$.
\WHILE{$v$ detects that $\quorum_v$ is occupied by the nodes not in its one-hop neighborhood.}
\STATE $v$ sets a positive natural number $k$ to be $k +=1$;
\STATE $v$ sets itself to be active in periods $T_i$, where $i=0,k\interfpara, 2k\interfpara,\cdots$;
\ENDWHILE
\end{algorithmic}
\end{algorithm}

By Algorithm~\ref{algorithm:quorum share}, each pair of neighboring nodes are physically connected even when the time slot allocation algorithm is implemented in Algorithm~\ref{algorithm:slot set assignment}.
According to Algorithm~\ref{algorithm:quorum share}, $u$ would be active in $\period_i$, where $i=0, k\interfpara, 2k\interfpara, \cdots$.
It means that $u$ is active in part of periods.
Since the grid QS satisfies the rotation closure property, $u$ is still physically connected with its neighboring nodes.
Therefore, we can obtain Lemma~\ref{lemma:quorum share physically connected}.
\begin{lemma}\label{lemma:quorum share physically connected}
A pair of neighboring nodes can be physically connected by Algorithm~\ref{algorithm:quorum share} under asynchronization.
\end{lemma}
\begin{proof}
Suppose a pair of neighboring nodes $u$ and $v$ respectively select quorums $\quorum_u$ and $\quorum_v$.
Thus $\quorum_u\cap\quorum_v\neq\varnothing$.
$\quorum_u$ and $\quorum_v$ respectively contain the time slot set $\subperiod_u$ and $\subperiod_v$.
Denote the relative clock shift  between $u$ and $v$ is $\clockshift^{uv}$.
Thus $\subperiod_u'=\subperiod_u+\{\clockshift^{uv}\}$.
Since the grid QS, $\subperiod_u'\cap\subperiod_v\neq\varnothing$.
\end{proof}

In Algorithm~\ref{algorithm:quorum share}, each node $u$ would be active in part of periods under asynchronization.
We find that the whole network delay is prolonged while each pair of neighboring nodes are guaranteed to be physically connected.
The demand of $u$ would be implemented lingeringly because the quorum $\quorum_u$ of $u$ moves to $\quorum'_u$ in another period as shown in Figure~\ref{fig:quorum shift}.
We suppose $\quorum'_u$ locates in the period $\period_1$ and $\period_1$ is originally assigned to the nodes in the region $\region_1$.
%So $u$ will compete the quorum with the nodes in $\region_1$.
So $u$ would share the same quorum with some nodes in $\region_1$.
At the worst case, all nodes in $\region_1$ shunt one turn, \ie they are active in $0, 2\interfpara, 4\interfpara, \cdots$.
The delay caused by the clock shift is at most $\clockshift^u$ periods when there is only $u$ which has clock shift under the interference model $\interferenceM$.
When the clocks of every nodes shift, their quorum may also shift.
For example, $\quorum_i$ and $\quorum_j$ respectively move to new places, such as $\quorum'_i$ and $\quorum'_j$, which locate in different periods in Figure~\ref{fig:quorum shift}.
%Since the clock of each node may shift, the quorums in $\region_2$ and $\region_3$ may also shift.
The worst case is that $\interfpara-1$ regions shift into one regions, thus the additional delay is at most $\interfpara(\interfpara-1)\periodsize$.
According to Theorem~\ref{theorem: the delay boundary}, we have the following lemma.
\begin{lemma}
By Algorithm~\ref{algorithm:quorum share}, the delay on the data aggregation is at most $O(\interfpara^2\periodsize\graphr+\degreemax^2)$ when the clocks are asynchronous.
\end{lemma}
\begin{figure}[htb]
\centering \includegraphics[scale=1, bb=229 410 376 435]{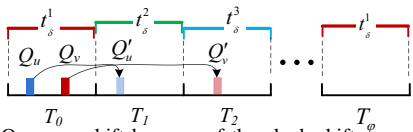}
\centering\caption{\footnotesize\label{fig:quorum shift}Quorums shift because of the clock shift.}
\end{figure} 
%%%%%%%%%%%%%%%%%%%%%%%%%%
%\subsection{Delay Performance under the Clock Asynchronization}
%\label{sec:asynchronous clock}
%\input{asynchronous_clock.tex}
%%%%%%%%%%%%%%%%%%%%%%%%%%
\section{Performance Evaluation}
\label{sec:experiment}
This paper evaluated \protocol and B-MAC in  a real testbed running TinyOS on TelosB motes.
The testbed composes of one hundred nodes.
%The testbed also composes of ten mini-HUB and 110 USB-ports.
%Thus each node can be supervised through the USB ports.
%In our experiment,
We compare the performance \protocol against B-MAC on the throughput and packet receiving ratio (PRR).
%\begin{figure}\hspace{0.0cm}
%\begin{minipage}[t]{0.45\linewidth}
%  % Requires \usepackage{graphicx}
%\centering\includegraphics[height=2.55cm,scale=0.38, bb=149 310 436 531]{figures/telosb.eps}\\
%\caption{The TelosB nodes.}\label{fig:telosb}
%\end{minipage}
%\hspace{0.1cm}\begin{minipage}[t]{0.45\linewidth}
%  % Requires \usepackage{graphicx}
%\centering\includegraphics[scale=0.38, bb=150 326 455 515]{figures/testbed.eps}\\
%\caption{The testbed.}\label{fig:deployment}
%\end{minipage}
%\end{figure}
\subsection{Experiment Setup}
We randomly deployed 100 nodes on an in-door test-bed.
Each sensor node works with its modified internal antenna, the transmission range of which can be as small as 10cm.
Thus nodes in the original network can still communicate with each other by multi-hop.
% and obtained an original real network.
After deployment, we start our experiment, composing of two phases.
%At the first phase, we established a tree by BFS and selected CDS, then colored the region according to Algorithm~\ref{algorithm:region coloring} and assigned slot set according to Algorithm~\ref{algorithm:slot set assignment}.
%Notice each node sets its power according to the method in~\cite{zhang2009transmission} in order to ensure the network connectivity.
At first phase, all the nodes are initially set with 100\% duty cycle.
At the second phase, nodes in a same region selected their quorums according to their locations (parent or leaf node) in the tree and the number of the leaf nodes under \protocol.
The duty cycle is set 20\% under B-MAC.

Under \protocol, each $\QS$ contains 100 time slots, \ie $\periodsize=100$.
Each time slot is respectively set as $50ms$, $1s$, $2s$ and $5s$.
Each node samples data in every $20ms$, $50ms$, $100ms$, $200ms$, $300ms$, $500ms$, $800ms$, $1s$, $1.5s$ and $2s$, which are called as the data generation period in Figure~\ref{fig:throughput} and~\ref{fig:prr}.
When the experiment starts, the sink broadcasts a message to synchronize the clocks of all nodes.

\subsection{Performance Comparison}
In this section, we compare the performance of B-MAC and \protocol on the network throughput and PRR.
Although we care about the channel utility, fairness and energy consumption, the experiment results on two parameters, throughput and PRR, synthetically reflect the channel utility.

\textbf{Throughput}. Figure~\ref{fig:throughput} shows the network throughput respectively under \protocol and B-MAC.
Each node except the sink generates data at different rate.
Because the nodes should compete the channel access when transmitting each packet, much time is wasted.
When the time slot size is big, for example, 1s, 2s and 5s, the nodes should cost time on the channel access competition and any pair of neighboring nodes have much continual time to transmit packets.
As shown in Figure~\ref{subfig:slot1000}, \ref{subfig:slot2000} and \ref{subfig:slot5000},  the throughput under B-MAC is much lower than that under \protocol when the time slot size is bigger, such as $1s$, $2s$ and $5s$.
Although the network is synchronized at the right beginning of  the experiment, the clocks of all nodes shift off after a period of time.
Some of nodes scheduled to wake up at common time may mismatch especially when time slots are set to be very short.
The throughput under B-MAC is litter higher than that under \protocol when the time slots size is small, such as $50ms$.
Notice that the network throughput under \protocol does not change much when the time slot size changes.
But the time slot size has much effect on the throughput under B-MAC.
The throughput under both B-MAC and \protocol decrease with the increasing of the data generation period when the period is higher than $900ms$.
\begin{figure*}[htp]
\centering\subfigure[50ms]{\includegraphics[angle=-90,scale=.24, bb=56 31 435 483]{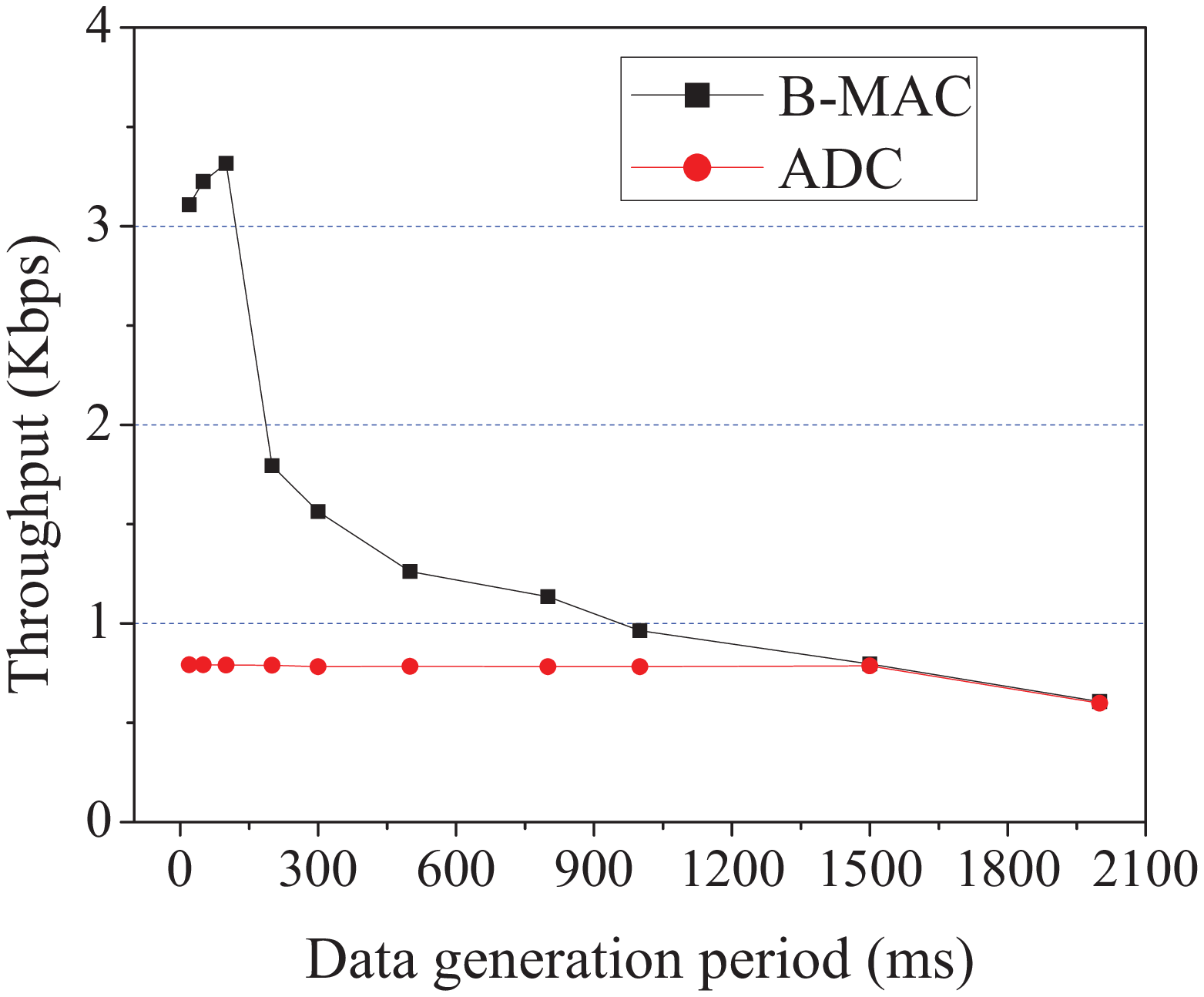}\label{subfig:slot50}}
\hspace{0.1cm}\subfigure[1000ms]{\includegraphics[angle=-90,scale=.24, bb=56 31 435 483]{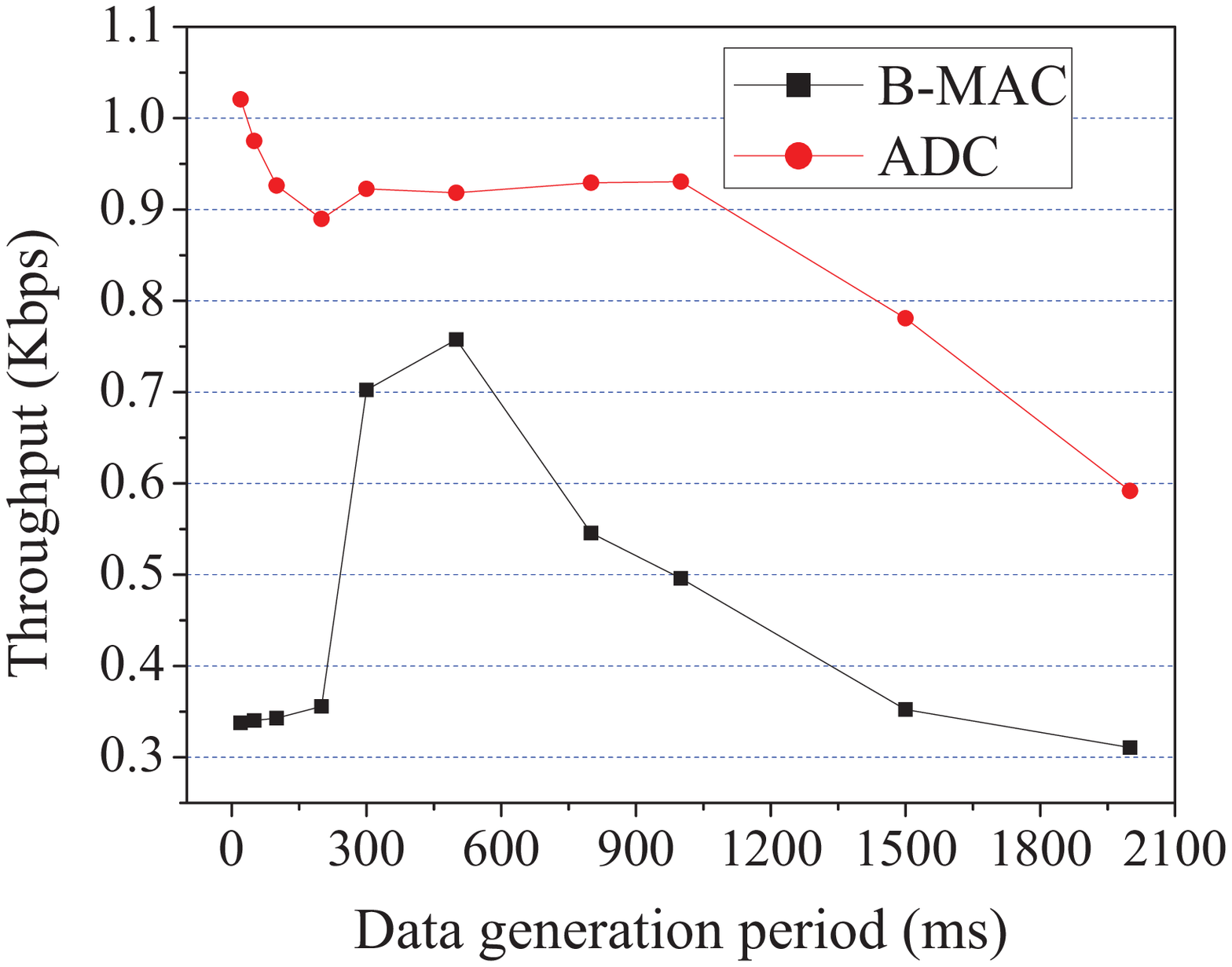}\label{subfig:slot1000}}
\hspace{0.1cm}\subfigure[2000ms]{\includegraphics[angle=-90,scale=.24, bb=56 31 435 483]{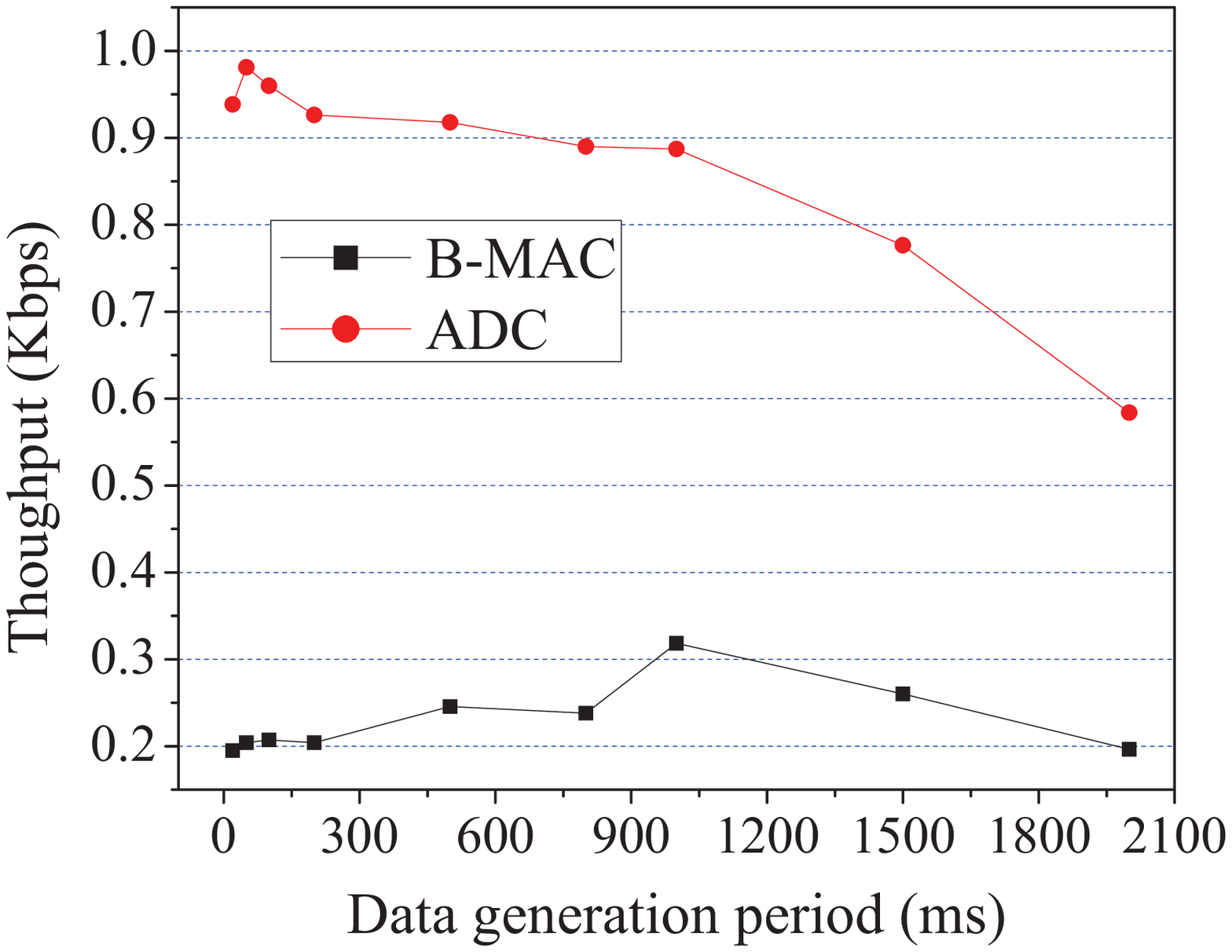}\label{subfig:slot2000}}
\hspace{0.1cm}\subfigure[5000ms]{\includegraphics[angle=-90,scale=.24, bb=56 31 435 483]{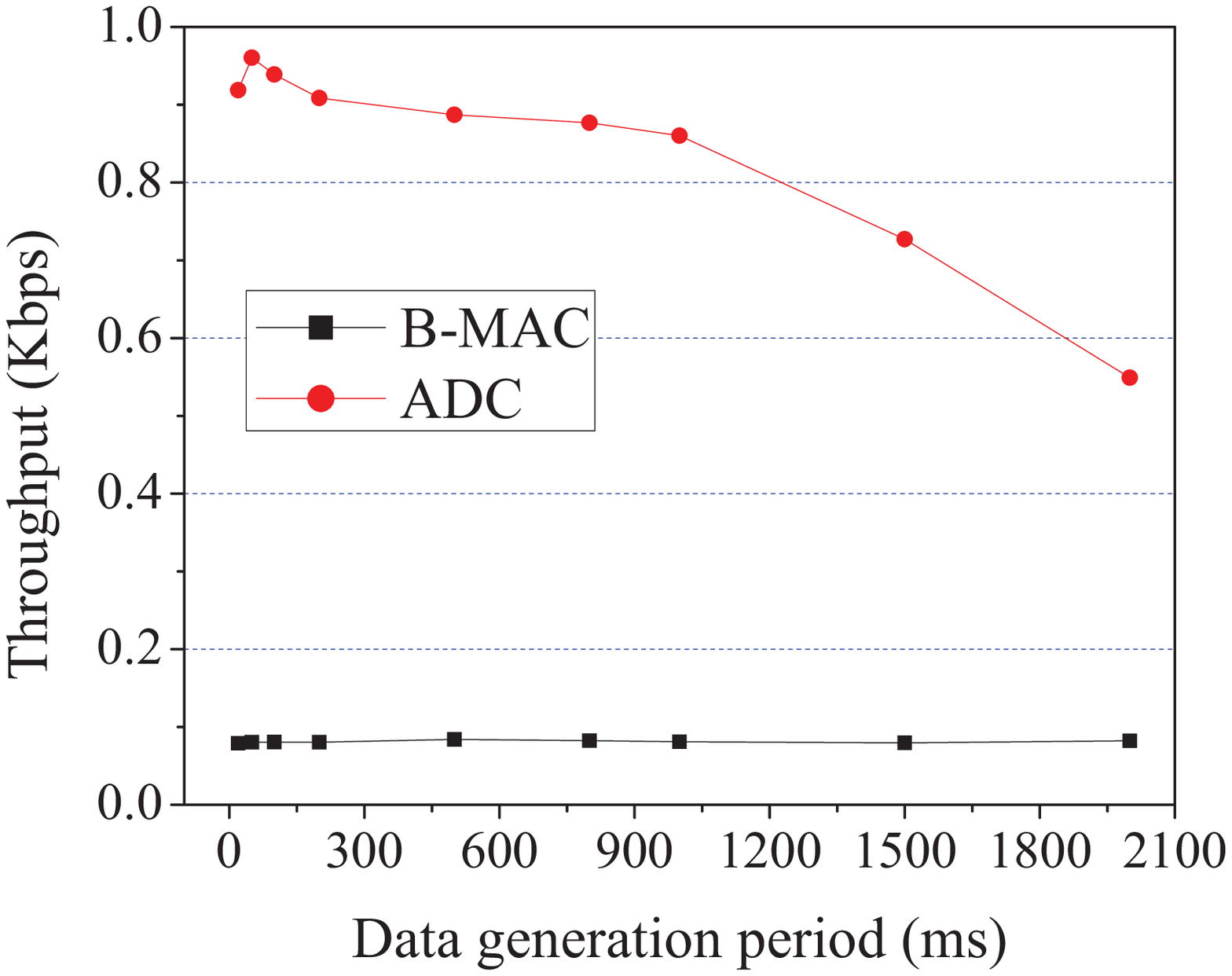}\label{subfig:slot5000}}
\centering\caption{\footnotesize\label{fig:throughput} The network throughput respectively under \protocol and B-MAC with different data generation periods.}
\end{figure*}

\textbf{PRR}. PRR reflects the channel utility within a network.
The PRRs under both B-MAC and \protocol increase with the increasing of the data generation period.
When the time slot size is big, such as, $1s$, $2s$ and $5s$, the PRR under \protocol is much higher than that under B-MAC.
The results are similar to those on the network throughput.
The time slots size has much effect on the throughput under B-MAC instead of that under \protocol.
The PRR under B-MAC increases with the decreasing of the time slot size.
\begin{figure*}[htp]
\centering\subfigure[50ms]{\includegraphics[angle=-90,scale=.23, bb=55 1 554 489]{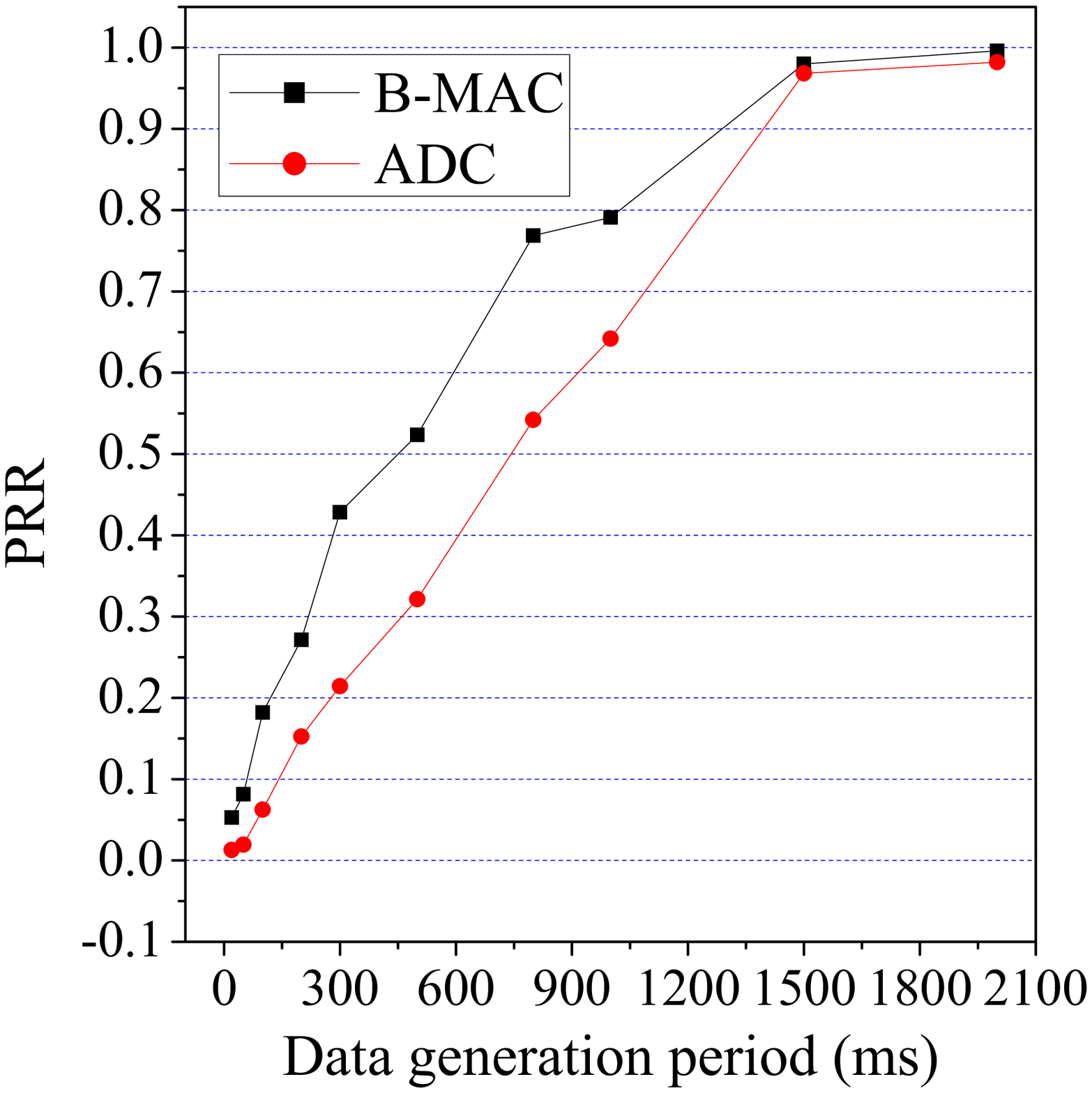}\label{subfig:prrslot50}}
\hspace{0.0cm}\subfigure[1000ms]{\includegraphics[angle=-90,scale=.23, bb=55 1 554 489]{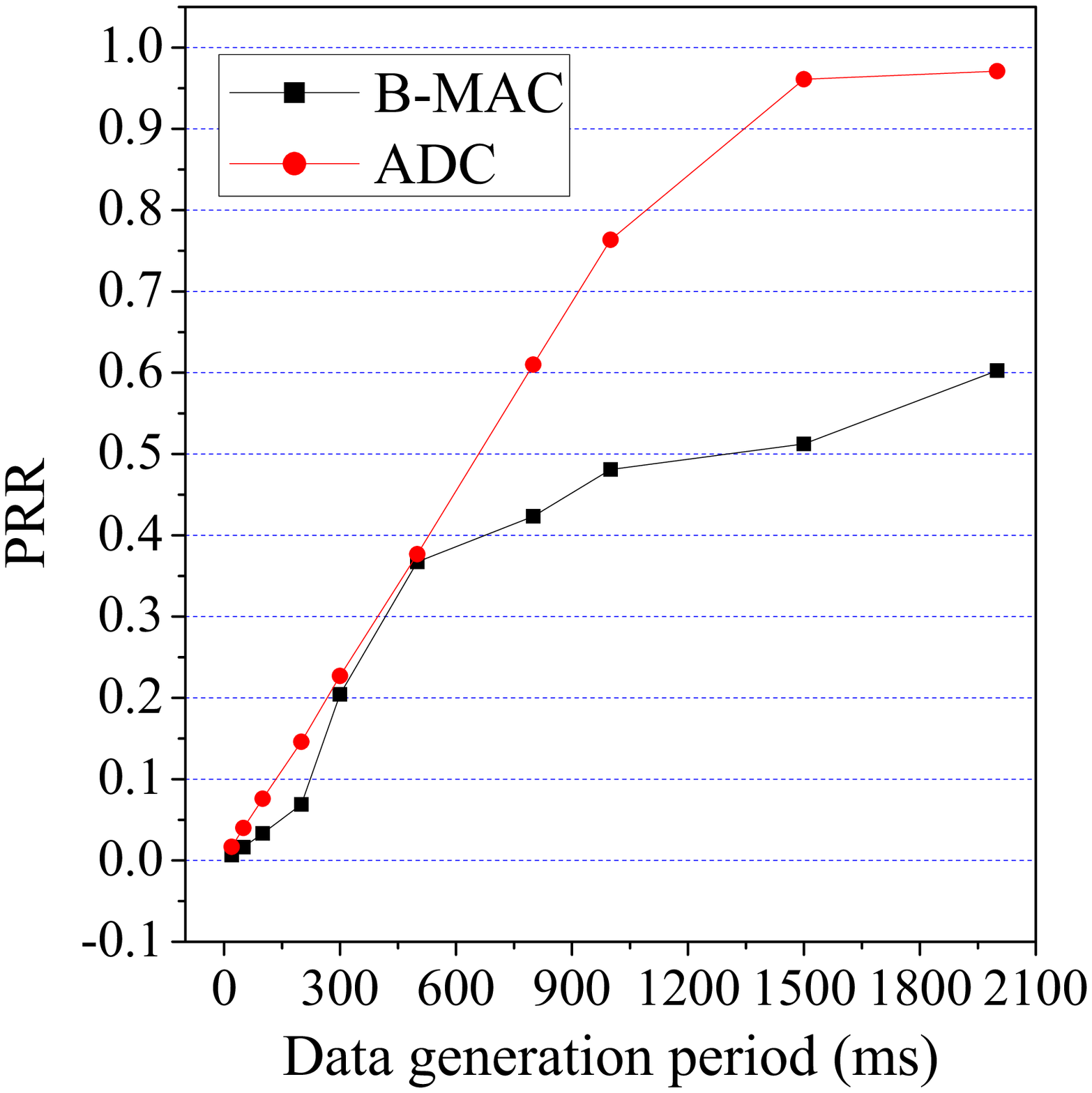}\label{subfig:prrslot1000}}
\hspace{0.0cm}\subfigure[2000ms]{\includegraphics[angle=-90,scale=.23, bb=55 1 554 489]{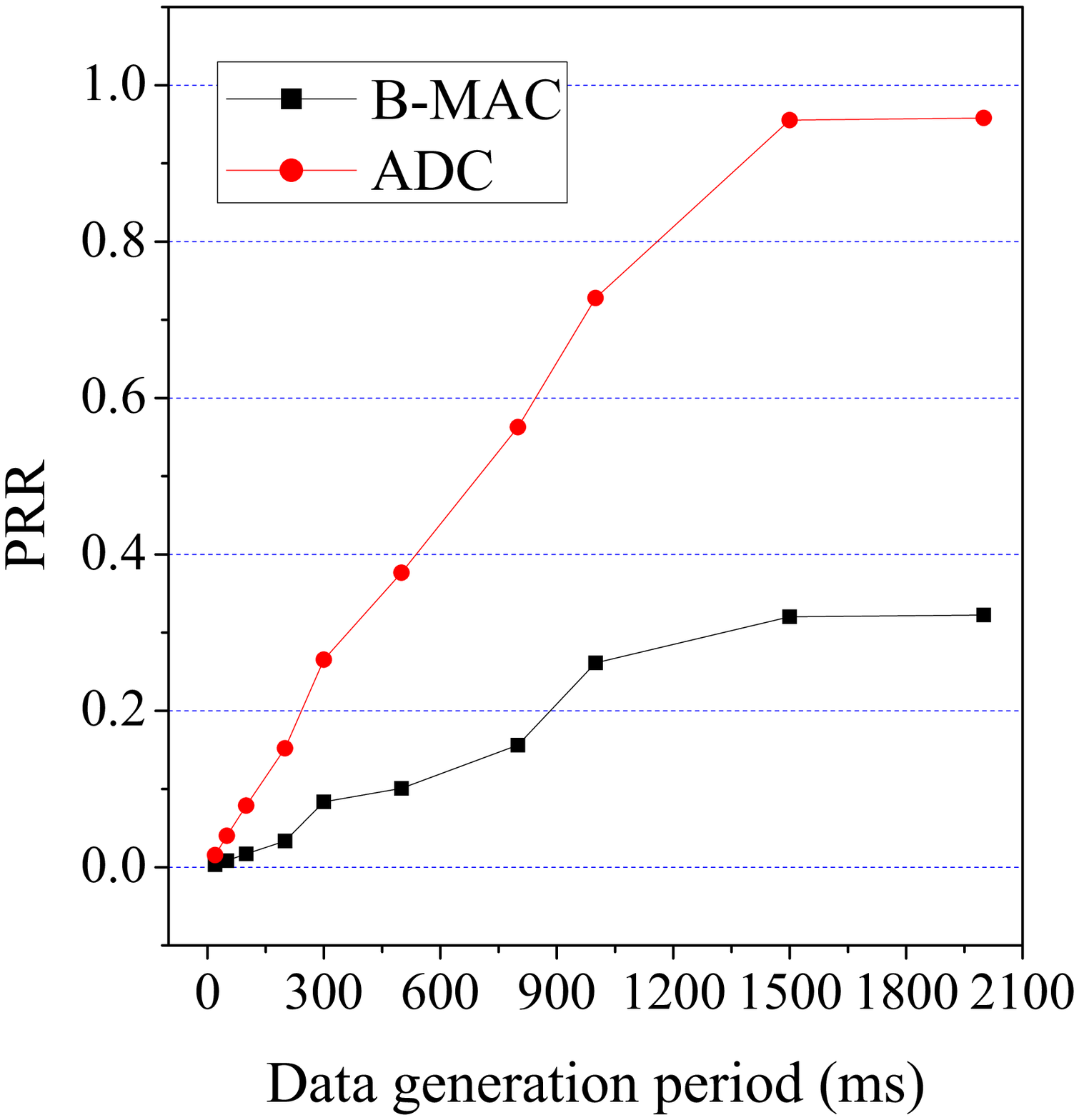}\label{subfig:prrslot2000}}
\hspace{0.0cm}\subfigure[5000ms]{\includegraphics[angle=-90,scale=.23, bb=55 1 554 489]{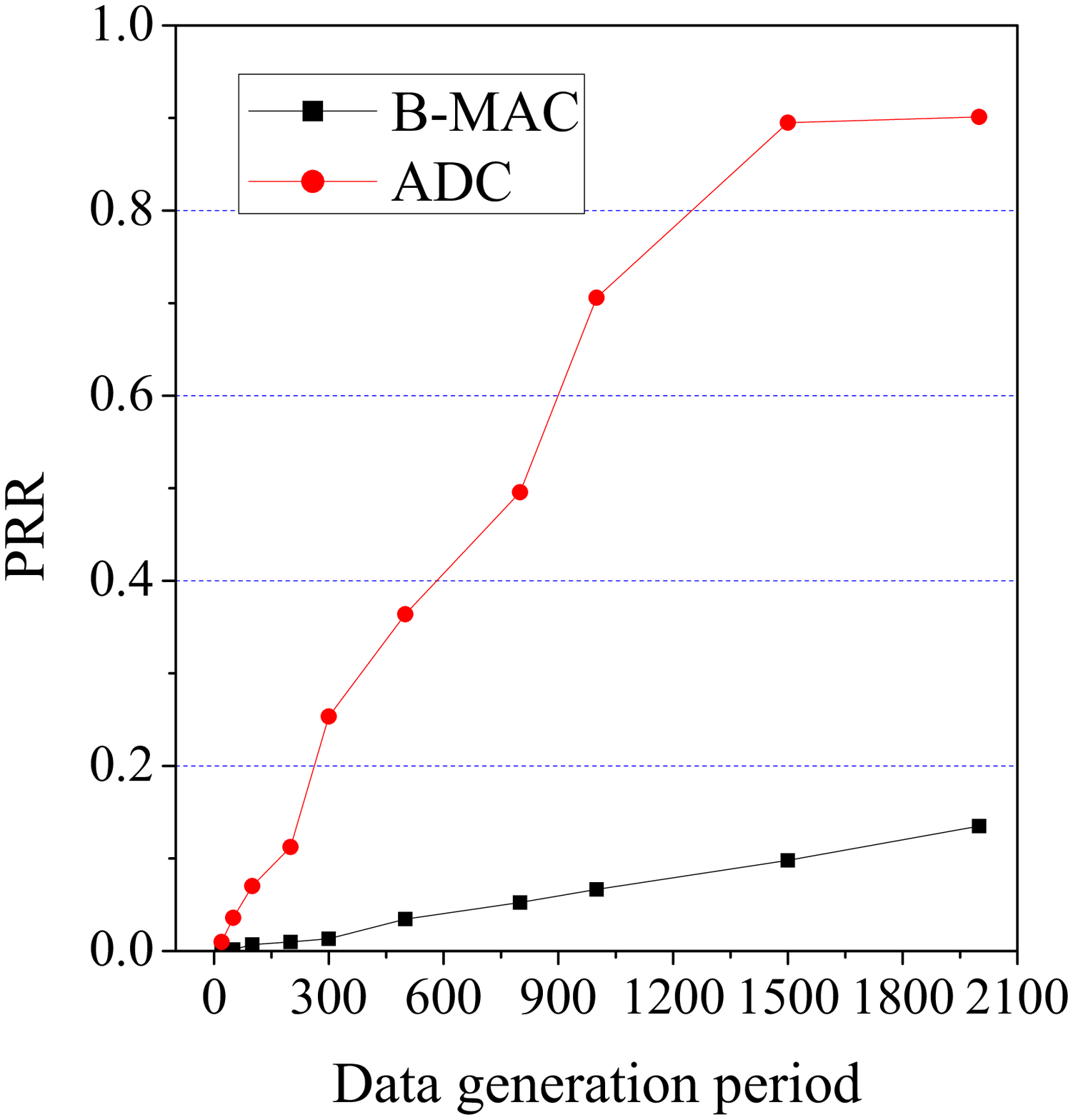}\label{subfig:prrslot5000}}
\centering\caption{\footnotesize\label{fig:prr} The network PRR respectively under \protocol and B-MAC with different data generation periods.}
\end{figure*}
%%%%%%%%%%%%%%%%%%%%%%%%%%%%%%%
\section{Related Work}
\label{sec:related work}
\subsection{Duty Cycle}
In WSNs, many works were put on duty-cycled networks as following.
\cite{gu2007data} designed a data forwarding technique to optimize the data delivery ration, end-to-end delay or energy consumption under low-duty-cycle by synchronized mode.
\cite{guo2009opportunistic} designed an opportunistic flooding scheme for low-duty-cycle networks with unreliable wireless links and predetermined wording schedules by locally synchronization.
\cite{wang2009duty} provided a benchmark for assessing diverse duty-cycle-aware broadcast strategies and extend it to distributed implementation.
\cite{hong2010minimum} minimized broadcast transmission delay by a set-cover-based approximation scheme with both centralized and distributed algorithms.
Using the $\beta$-synchronizer, a fast distributed algorithm built all-to-one shortest paths with polynomial message and time complexity~\cite{lai2010distributed}.
\cite{sun2009adb} designed an asynchronous duty-cycle broadcasting to let a node be active very long time when it need broadcast the data to a large number of neighbors.
~\cite{nath2007communicating} analyzed the performance of geographic routing over duty-cycled nodes and presented a sleeping scheduling algorithm that can be tuned to achieve a target routing latency.
%~\cite{jiao2010duty} investigates the duty-cycle-aware minimum latency broadcast scheduling (DCA-MLBS) problem in multi-hop wireless networks and proposed an approximation algorithm called OTAB for one-to-all DCA-MLBS problem and two for the all-to-all DCA-MLBS problem, called UTB and UNB.
%The OTAB algorithm achieves a constant approximation ratio of $17|T|+1$, where $|T|$ is the number of time-slots in a scheduling period.
%The UTB and UNB algorithms achieve the approximation ratios of $17|T|+2$ and $(\Delta+22)|T|+1$ respectively, where $\Delta$ is the maximum nodes degree of the network.
%~\cite{hsin2004network} examined the fundamental relationship between the reduction in sensor duty cycle and the required level of redundancy for a fixed performance measure and explored the design of good sensor sleep schedules under the random sleep and coordinated sleep types.
~\cite{benson2008opportunistic} presented an alternative frame-let based LPL implementation to improve the network performance by opportunistically aggregating packets over the radio channel.
\subsection{MAC protocol}
%Medium access control (MAC) protocols have been developed to allocate the channels or to assist multiple nodes to access the channel.
%MAC protocols can be broadly divided into two groups: scheduled and contention-based~\cite{ye2004medium}.
%TDMA belongs to the first group and CSMA and ALOHA belongs to second group.
%One class of MAC protocols is to avoid interference by scheduling nodes into different sub-channels which are divided either by time, frequency or orthogonal codes, such as TDMA, Frequency Division Multiple Access (FDMA) and Code Division Multiple Access (CDMA).
%Another class of MAC protocols is contention-based, such as CSMA and ALOHA.
%
%\cite{jian2008can} demonstrates that CSMA/CA networks exhibit severe fairness problem in many scenarios.
%Many fairness solutions were designed to solve the problem in recent years and can be classified into overhearing-based~\cite{huang2001max} and non-overhearing~\cite{heusse2005idle}.
%In the overhearing-based solutions, each node should consume much resource on the information collection and the decision of its own media contention policy~\cite{yang2005physical}.
%The non-overhearing-based solutions use the classical Additive Increase Multiplicative Decrease (AIMD) for rate control.
%But unsynchronized multiplicative decrease causes the failure of AIMD in CAMA/CA networks.
%To increase the spatial reuse, some works on MAC are designed by tuning the carrier-sense threshold~\cite{yang2005physical,polastre2004versatile} or the transmission power~\cite{zhang2009transmission}.
%
Some protocols were designed to combine the advantages of TDMA and CSMA.
\cite{rhee2008z} proposed a hybrid MAC protocol, called Z-MAC, in which, a node always performs carrier-sensing before transmission.
Z-MAC consumes much energy on the carrier-sensing and also needs local synchronization among senders in two-hop neighborhoods.
S-MAC~\cite{ye2004medium} and T-MAC~\cite{van2003adaptive} employ RTS/CTS mechanism to solve the the synchronization failure.
Since these protocols use RTS/CTS, their overhead is quite high~\cite{rhee2008z}.
B-MAC~\cite{polastre2004versatile} is the default MAC in the operate system of Mica2 and adopts Low Power Listening (LPL) to solve the asynchronization.
Since LPL consumes much energy, X-MAC reduces the energy consumption and latency by employing short preamble and embedding address information of the target in the preamble~\cite{buettner2006x}.
So the non-target receivers can quickly go back to sleep and the energy is saved.
LPL based preamble transmission may occupy the medium for much longer time than actual data transmission.
So~\cite{sun2008ri} designed an asynchronous duty cycle MAC: RI-MAC.
It wastes energy especially under low traffic load and the interference is increased because of the periodical broadcasting of beacons.

MAC protocols are also designed to reduce energy consumption, such as S-MAC~\cite{ye2004medium} and T-MAC~\cite{van2003adaptive}.
\cite{zheng2003asynchronous}  considered LPL approaches, such as WiseMAC and B-MAC, are limited to duty cycles of 1-2\% and  designed a new MAC protocol called scheduled channel polling (SCP) to ensure that duty cycles of 0.1\% and below are possible.
\cite{zheng2003asynchronous} dynamically adjusts duty cycles in the face of busy networks and streaming traffic in order to reduce the latency.
%But SCP should synchronizes the polling times of all neighboring nodes in order to eliminates long preambles in LPL for all transmissions.
%Notice that existing hybrid protocols do not consider fairness.
\cite{dutta2010design} presented a new receiver-initiated link layer A-MAC to support multiple services under a unified architecture more efficiently and scalably than prior designs.
\cite{raman2010pip} designed a TDMA-based MAC primitive module PIP to achieve high throughput for reliable bulk data transfer.
%%%%%%%%%%%%%%%%%%%%
\section{Conclusion}
\label{sec:conclusion}
%\input{conclusion.tex}
%In this paper, we have proposed a localized scheme, $\protocol$, to adaptively adjust the duty cycle
%according to its demand such that each node can fairly use the channel.
%%$\protocol$ combines the advantages of both TDMA and CSMA and don't need extra synchronization algorithms.
%When the network is asynchronous, $\protocol$ can still guarantee that any pair of neighboring nodes has a
% common active time to communicate with each other.
%%In $\protocol$, each node adjusts the duty cycle
%$\protocol$ uses the $K$-RBP to decrease the load of each quorum, thus  increasing the channel utilization.
%We implemented \protocol in TinyOS on a testbed of 100 TelosB nodes and evaluated it in a multihop  network.
%Compared to a TinyOS implementation of B-MAC, \protocol significantly improves the performance such as
%network throughput and PRR.
%There are many other MAC protocols, the performance of some of which is over that of B-MAC.
%In the future work, we will give the comparison between our scheme and the existing protocols since it is still an open problem that how to sleep the node when the duty cycle is extra low.
Energy conservation is a fundamental issue in WSNs, which usually relies on wise designs of duty cycling mechanisms. In this paper, we propose a localized scheme, \protocol, to adaptively adjust the duty cycles of all nodes in WSNs.
\protocol leverages the technique of QS and adjust the duty cycles of sensor nodes according to their demand, so that all nodes can fairly access their common channels.
We address both synchronous and asynchronous cases with \protocol and implement it on a test-bed with 100 TelosB nodes.
The results demonstrate that ADC significantly improves the WSN performance such as network throughput and PRR.
In our future work, we plan to design protocols of duty cycle adjustment, which has more high utilization of active time and lower duty cycle, so the energy consumption efficiency can be increased.
%%%%%%%%%%%%%%%%%%%%
%\section{Acknowledgement}
%\label{section:discussion}
%\input{discussion.tex}

%%%%%%%%%%%%%%%%%%%%%%%%%%%%%%%%%%%%%%
\bibliographystyle{unsrt}
\bibliography{pakeyref}

\begin{thebibliography}{10}

\bibitem{mo2009canopy}
L.~Mo, Y.~He, Y.~Liu, J.~Zhao, S.J. Tang, X.Y. Li, and G.~Dai.
\newblock {Canopy closure estimates with greenorbs: Sustainable sensing in the
  forest}.
\newblock In {\em ACM Sensys}, pages 99--112, 2009.

\bibitem{mar¨®ti2004flooding}
M.~Mar{\'o}ti, B.~Kusy, G.~Simon, and {\'A}.~L{\'e}deczi.
\newblock {The flooding time synchronization protocol}.
\newblock In {\em ACM Sensys}, pages 39--49, 2004.

\bibitem{ye2004medium}
W.~Ye and J.~Heidemann.
\newblock {Medium access control in wireless sensor networks}.
\newblock {\em Wireless sensor networks}, pages 73--91, 2004.

\bibitem{polastre2004versatile}
J.~Polastre, J.~Hill, and D.~Culler.
\newblock {Versatile low power media access for wireless sensor networks}.
\newblock In {\em ACM Sensys}, pages 03--05, November 2004.

\bibitem{van2003adaptive}
T.~Van~Dam and K.~Langendoen.
\newblock {An adaptive energy-efficient MAC protocol for wireless sensor
  networks}.
\newblock In {\em ACM Sensys}, page 180, 2003.

\bibitem{jian2008can}
Y.~Jian and S.~Chen.
\newblock {Can CSMA/CA networks be made fair?}
\newblock In {\em Proceedings of ACM MobiCom}, pages 235--246, 2008.

\bibitem{malkhi1998byzantine}
D.~Malkhi and M.~Reiter.
\newblock {Byzantine quorum systems}.
\newblock {\em Distributed Computing}, 11(4):203--213, 1998.

\bibitem{bian2009quorum}
K.~Bian, J.M. Park, and R.~Chen.
\newblock {A quorum-based framework for establishing control channels in
  dynamic spectrum access networks}.
\newblock In {\em Proceedings of ACM MobiCom}, pages 25--36, 2009.

\bibitem{wu2008fully}
S.H. Wu, M.S. Chen, and C.M. Chen.
\newblock {Fully adaptive power saving protocols for ad hoc networks using the
  Hyper Quorum System}.
\newblock In {\em IEEE ICDCS}, pages 785--792, 2008.

\bibitem{chaporkar2006dynamic}
P.~Chaporkar, S.~Sarkar, and R.~Shetty.
\newblock {Dynamic quorum policy for maximizing throughput in limited
  information multiparty MAC}.
\newblock {\em IEEE/ACM Transactions on Networking (TON)}, 14(4):848, 2006.

\bibitem{li2008capacity}
S.~Li, Y.~Liu, and X.Y. Li.
\newblock {Capacity of large scale wireless networks under Gaussian channel
  model}.
\newblock In {\em Proceedings of ACM MobiHoc}, pages 140--151, 2008.

\bibitem{tseng2003power}
Y.C. Tseng, C.S. Hsu, and T.Y. Hsieh.
\newblock {Power-saving protocols for IEEE 802.11-based multi-hop ad hoc
  networks}.
\newblock {\em Computer Networks}, 43(3):317--337, 2003.

\bibitem{jiang2005quorum}
J.R. Jiang, Y.C. Tseng, C.S. Hsu, and T.H. Lai.
\newblock {Quorum-based asynchronous power-saving protocols for IEEE 802.11 ad
  hoc networks}.
\newblock {\em Mobile Networks and Applications}, 10(1):169--181, 2005.

\bibitem{naor1994load}
M.~Naor and A.~Wool.
\newblock {The load, capacity, and availability of quorum systems}.
\newblock In {\em ANNUAL SYMPOSIUM ON FOUNDATIONS OF COMPUTER SCIENCE},
  volume~35, pages 214--214, 1994.

\bibitem{li2009multiple}
X.Y. Li, Y.~Wang, and W.~Feng.
\newblock {Multiple Round Random Ball Placement: Power of Second Chance}.
\newblock In {\em INFOCOM}, page 448, 2009.

\bibitem{mitzenmacher2001power}
M.~Mitzenmacher.
\newblock {The power of two choices in randomized load balancing}.
\newblock {\em IEEE TPDS}, pages 1094--1104, 2001.

\bibitem{li2009Efficient}
ShiGuang Wang ShaoJie Tang GuoJun Dai JiZhong Zhao Yong~Qi Xiang-Yang~Li,
  XiaoHua~Xu.
\newblock {Efficient Data Aggregation in Multi-hop Wireless Sensor Networks
  under Physical Interference Model}.
\newblock In {\em Proceedings of IEEE MASS}, 2009.

\bibitem{WanHWWJ09}
Peng-Jun Wan, Scott C.-H. Huang, Lixin Wang, Zhiyuan Wan, and \emph{etc}.
\newblock Minimum-latency aggregation scheduling in multihop wireless networks.
\newblock In {\em Proceedings of ACM MobiHoc}, pages 185--194, 2009.

\bibitem{gu2007data}
Y.~Gu and T.~He.
\newblock {Data forwarding in extremely low duty-cycle sensor networks with
  unreliable communication links}.
\newblock In {\em ACM Sensys}, page 334. ACM, 2007.

\bibitem{guo2009opportunistic}
S.~Guo, Y.~Gu, B.~Jiang, and T.~He.
\newblock {Opportunistic flooding in low-duty-cycle wireless sensor networks
  with unreliable links}.
\newblock In {\em ACM MobiCom}, pages 133--144, 2009.

\bibitem{wang2009duty}
F.~Wang and J.~Liu.
\newblock {Duty-cycle-aware broadcast in wireless sensor networks}.
\newblock In {\em IEEE INFOCOM}, 2009.

\bibitem{hong2010minimum}
J.~Hong, J.~Cao, W.~Li, S.~Lu, and D.~Chen.
\newblock {Minimum-Transmission Broadcast in Uncoordinated Duty-Cycled Wireless
  Ad Hoc Networks}.
\newblock {\em IEEE TVT}, 59(1):307--318, 2010.

\bibitem{lai2010distributed}
S.~Lai and B.~Ravindran.
\newblock {On Distributed Time-Dependent Shortest Paths over Duty-Cycled
  Wireless Sensor Networks}.
\newblock In {\em IEEE INFOCOM}, 2010.

\bibitem{sun2009adb}
Y.~Sun, O.~Gurewitz, S.~Du, L.~Tang, and D.B. Johnson.
\newblock {ADB: an efficient multihop broadcast protocol based on asynchronous
  duty-cycling in wireless sensor networks}.
\newblock In {\em ACM Sensys}, pages 43--56, 2009.

\bibitem{nath2007communicating}
S.~et~al Nath.
\newblock {Communicating via fireflies: Geographic routing on duty-cycled
  sensors}.
\newblock In {\em ACM/IEEE IPSN}, pages 440--449, 2007.

\bibitem{benson2008opportunistic}
J.~Benson, T.~O'Donovan, U.~Roedig, and C.J. Sreenan.
\newblock {Opportunistic aggregation over duty cycled communications in
  wireless sensor networks}.
\newblock In {\em ACM/IEEE IPSN}, pages 307--318, 2008.

\bibitem{rhee2008z}
I.~Rhee, A.~Warrier, M.~Aia, J.~Min, and M.L. Sichitiu.
\newblock {Z-MAC: a hybrid MAC for wireless sensor networks}.
\newblock {\em IEEE/ACM TON}, 16(3):511--524, 2008.

\bibitem{buettner2006x}
M.~Buettner, G.V. Yee, E.~Anderson, and R.~Han.
\newblock {X-MAC: a short preamble MAC protocol for duty-cycled wireless sensor
  networks}.
\newblock In {\em Proceedings of ACM Sensys}, page 320, 2006.

\bibitem{sun2008ri}
Y.~Sun, O.~Gurewitz, and D.B. Johnson.
\newblock {RI-MAC: a receiver-initiated asynchronous duty cycle MAC protocol
  for dynamic traffic loads in wireless sensor networks}.
\newblock In {\em ACM Sensys}, pages 1--14, 2008.

\bibitem{zheng2003asynchronous}
R.~Zheng, J.C. Hou, and L.~Sha.
\newblock {Asynchronous wakeup for ad hoc networks}.
\newblock In {\em Proceedings of ACM MobiHoc}, pages 35--45. ACM, 2003.

\bibitem{dutta2010design}
P.~Dutta, S.~Dawson-Haggerty, Y.~Chen, C.J.M. Liang, and A.~Terzis.
\newblock {Design and evaluation of a versatile and efficient
  receiver-initiated link layer for low-power wireless}.
\newblock In {\em ACM Sensys}, pages 1--14, 2010.

\bibitem{raman2010pip}
B.~Raman, K.~Chebrolu, S.~Bijwe, and V.~Gabale.
\newblock {PIP: a connection-oriented, multi-hop, multi-channel TDMA-based MAC
  for high throughput bulk transfer}.
\newblock In {\em ACM Sensys}, pages 15--28, 2010.

\end{thebibliography}
\end{document}